%% file: main.tex
\documentclass[a4paper,onecolumn,accepted=2024-10-22]{quantumarticle}
\pdfoutput=1

\usepackage[utf8]{inputenc}
\usepackage{amsmath,amsthm,amsfonts,amssymb,mathrsfs,thmtools,thm-restate}
\usepackage{bbm}
\usepackage{csquotes}
\usepackage{graphicx}
\graphicspath{ {./images/} }
\usepackage{mathtools}
\usepackage{xcolor}
\usepackage{cuted}
\usepackage{tikz}
\usepackage{authblk}
\usepackage{mdframed}
\usepackage{hyperref}
\usetikzlibrary{quantikz}
\usepackage[bibencoding=auto,style=alphabetic,defernumbers=true,backend=biber,giveninits=true,doi=true,isbn=false,url=false]{biblatex}
\usepackage{geometry}
\mdfsetup{%
middlelinewidth=2pt,
backgroundcolor=orange!10,
roundcorner=10pt
}

\newcommand{\marten}[1]{}
\newcommand{\niels}[1]{}
\newcommand{\bruno}[1]{}
\newcommand{\harry}[1]{}

\newtheorem{theorem}{Theorem}[section]
\newtheorem{lemma}[theorem]{Lemma}
\newtheorem{corollary}[theorem]{Corollary}
\newcounter{mycount}

\theoremstyle{definition}
\newtheorem{definition}[theorem]{Definition}
\newtheorem{notation}[theorem]{Notation}
\newtheorem{example}[theorem]{Example}

\theoremstyle{plain}
\newtheorem{remark}[theorem]{Remark}

\DeclareMathOperator\BQP{\mathsf{BQP}}
\DeclareMathOperator\NC{\mathsf{NC}}
\DeclareMathOperator\AC{\mathsf{AC}}
\DeclareMathOperator\TC{\mathsf{TC}}
\DeclareMathOperator\QNC{\mathsf{QNC}}

\DeclareMathOperator\LAQCC{\mathsf{LAQCC}}

\newcommand{\ceil}[1]{\left\lceil #1 \right\rceil}
\newcommand{\kb}[1]{\ket{#1}\bra{#1}}
\newcommand{\mathO}{\mathcal{O}}

\title{State preparation by shallow circuits using feed forward}

\author[1]{Harry Buhrman}
\author[1]{Marten Folkertsma}
\author[2]{Bruno Loff}
\author[1,3]{Niels M. P. Neumann}
\date{}

\affil[1]{QuSoft, CWI \& University of Amsterdam, Amsterdam, the Netherlands}
\affil[2]{LASIGE \& Department of Mathematics, University of Lisbon}
\affil[3]{The Netherlands Organisation for Applied Scientific Research (TNO), Delft, the Netherlands}

\begin{document}
\maketitle
\begin{abstract}
Fault tolerant quantum computers repetitively apply a four-step procedure: 
First, perform a few one and two-qubit quantum gates. 
Second, perform a syndrome measurement on a subset of the qubits. 
Third, perform fast classical computations to establish if and where errors occurred. 
And, fourth, correct the errors with a correction step. 
The next iteration applies the same procedure with new one and two-qubit gates. 
Even though current error-rates prohibit this procedure to work and fault tolerant quantum computing remains a distant goal, the same procedure can already prove useful today. 
In this work we make use of this four-step scheme not to carry out fault-tolerant computations, but to enhance short, {\em constant}-depth, quantum circuits that perform 1 qubit gates and {\em nearest-neighbor} 2 qubit gates.

We introduce a new computational model called \emph{Local Alternating Quantum Classical Computations} ($\LAQCC$). 
In this model, qubits are placed in a grid and they can only interact with their direct neighbors; the quantum circuits are of constant depth with intermediate measurements; a classical controller can perform log-depth computations on these intermediate measurement outcomes and control future quantum operations based on the outcome.
This model fits naturally between quantum algorithms in the NISQ era and full-fledged fault-tolerant quantum computation. 
We first prove that any Clifford circuit has an equivalent $\LAQCC$ circuit, and that any $\LAQCC$ circuit can be simulated by a $\QNC^1$ circuit. 
Next, we conjecture the non-simulatability of $\LAQCC$ by showing that $\LAQCC$ contains the class of Instantaneous Quantum Polynomial-time circuits.
We also show that any $\LAQCC$ circuit with polynomial-sized quantum circuits and unbounded classical computations is contained in the class of quantum circuits equipped with post-selection gates with respect to the task of state preparation.
We continue by presenting $\LAQCC$ implementations of different subroutines, including OR-gates, quantum Fourier transforms and Threshold gates. 
These subroutines prove vital in constructing three state preparation routines in the main part of this work. 
Preparing a uniform superposition uses constant-depth arithmetic gates, combined with an exact Grover implementation by Long. 
For the $W$-state, we employ a compress-uncompress method to switch between a binary and one-hot encoding. 
This method extends to the more generalized Dicke-states, the superposition of $n$-bit strings of Hamming weight $k$, for $k=\mathO(\sqrt{n})$, but fails for higher $k$ due to the birthday paradox. 
We extend this protocol to a protocol that prepares many-body scar states, highly excited states with low entanglement and longer coherence times than states with the same energy density. 
We present a circuit for preparing Dicke-states for larger $k$ requiring log-depth circuits that maps between the factoradic number system and the combinatorial number system. 
\end{abstract}

\input{introduction.tex}

\input{preliminaries.tex}

\input{model.tex}

\section{State preparation in \texorpdfstring{$\LAQCC$}{LAQCC}}\label{sec:state_prep_in_LAQCC}
In this section we consider what quantum states we can prepare using an $\LAQCC$ circuit beyond the stabilizer states and Clifford circuits discussed in the previous section. 
Specifically, as mentioned in the introduction, we consider quantum states that are widely used, in other quantum algorithms, for bench marking purposes and within physics. 
First, we show how to create a uniform superposition of computational basis states up to size $q$, where $q$ is not a power of $2$, a state that is often used as initial state in other algorithms (including the following other state preparation protocols presented in this work). 
We then use this procedure to create $W$-states, the uniform superposition over all $n$-bitstrings of Hamming-weight $1$, using a compress-uncompress method. 
This compress-uncompress method generalizes to preparing Dicke-$(n,k)$ states for $k=\mathO(\sqrt{n})$, uniform superpositions over all $n$-bitstrings of Hamming-weight $k=\mathO(\sqrt{n})$. 
Dicke states find many applications, and especially the compress-uncompress approach might prove useful for entanglement distillation protocols.
Preparing general Dicke-$(n,k)$ states requires a novel method to map between two integer representation systems, the factoradic representation and the combinatorial number representation. 
Finally, we present a state preparation protocol for quantum many-body scar states, states often used in physics, based on the Dicke-$(n,k)$ state preparation protocol for $k=\mathO(\sqrt{n})$

\input{Uniform_superposition.tex}
\input{W_states.tex}

\input{Dicke_states.tex}

\section*{Acknowledgements}
We want to thank Jonas Helsen, Joris Kattem{\"o}lle, Ido Niesen, Kareljan Schoutens, Florian Speelman, Dyon van Vreumingen and Jordi Weggemans for insight full discussions. 
Furthermore, we would like to thank Georgios Styliaris for first mentioning how to parallelize Clifford ladder circuits. 
HB and MF were supported by the Dutch Ministry of Economic Affairs and Climate Policy (EZK), as part of the Quantum Delta NL programme. BL was funded by the European Union (ERC, HOFGA, 101041696). Views and opinions
expressed are however those of the author(s) only and do not necessarily reflect those of the European Union or the European Research Council. Neither the European Union nor the granting authority can be held responsible for them. Also supported by FCT through the LASIGE Research Unit, ref. UIDB/00408/2020 and ref. UIDP/00408/2020. NN was supported by the quantum application project of TNO. This work was supported by the Dutch Research Council (NWO/OCW), as part of the Quantum Software Consortium programme (project number 024.003.037).

\printbibliography

\appendix

\input{appendix.tex}

\end{document}

%% file: introduction.tex
\section{Introduction}
Current quantum hardware is unable to carry out universal quantum computations due to the buildup of errors that occur during the computation. 
The magnitude of the individual error is currently above the value that the Threshold Theorem requires in order to kick-start quantum error correction and fault-tolerant quantum computation~\cite[Section 10.6]{nielsen_chuang_2010}. 
Although the experimentally achieved fidelity rates are promising and the error bounds are inching closer to the required threshold, we will have to work for the foreseeable future with quantum hardware with errors that build-up during the computation.  This implies that we can only do a limited number of steps before the output of the computation has become completely uncorrelated with the intended one.

For fault-tolerant quantum computing, we repeat four steps: 
1) We apply a number of single and two-qubit quantum gates, in parallel whenever possible; 
2) We perform a syndrome measurement on a subset of the qubits; 
3) We perform fast classical computations to determine which errors have occurred and how to correct them; 
and, 4) We apply correction terms based on the classical computations.
We then repeat these four steps with a next sequence of gates. 
These four steps are essential to fault-tolerant quantum computing.

The starting point of this work is to use the four steps outlined above, not to carry out error correction and fault-tolerant computation, but to enhance short, constant-depth, {\em uncorrected} quantum circuits that perform single qubit gates and {\em nearest-neighbor} two qubit gates. 
Since in the long run we will have to implement error-correction and fault-tolerant computation anyhow, and this is done by such a four-step process, why not make other use of this architecture? Moreover, on some of the quantum hardware platforms, these operations are already in place.
Embracing this idea we naturally arrive at the question: what is the computational power of \textit{low-depth} quantum-classical circuits organized as in the four steps outlined above? 
We thus investigate circuits that execute a small, ideally constant, number of stages, where at each stage we may apply, in parallel, single qubit gates and {\em nearest-neighbor} two qubit gates, followed by measurements, followed by low-depth classical computations of which the outcome can control quantum gates in later stages. 
It is not clear, at first, whether such circuits, especially with constant depth, can do anything remotely useful. 
But we will see that this is indeed the case: many quantum computations can be done by such circuits in constant depth. 
By parallelizing quantum computations in this way, we improve the overall computational capabilities of these circuits, as we do not incur errors on qubits that are idle, simply because qubits are not idle for a very long time. 
Furthermore, reducing the depth of quantum circuits, at the cost of increasing width, allows the circuit to be run faster even if errors occur.

The first usage of such a four-step layout, not to do error correction, but to perform computations, can be found in the paradigm of measurement-based quantum computing~\cite{gottesman1999demonstrating,raussendorf2001one,jozsa2006introduction,clark2007generalised}: 
A universal form of quantum computing where a quantum state is prepared and operations are performed by measuring qubits in different bases, depending on previous measurements and intermediate measurements.

\citeauthor{PhamSvore2013} were the first to formalize the four-step protocol for performing computations~\cite{PhamSvore2013}. They included specific hardware topologies by considering two-dimensional graphs for imposing constraints on qubit interactions. In their model, they develop circuits for particularly useful multi-qubit gates, including specifying costs in the width, number of qubits, depth, number of concurrent time steps, size, and total number of non-Identity operations.
As a result, they find an algorithm that factors integers in polylogarithmic depth.
\citeauthor{Browne:2011} showed that the main tool in the work by \citeauthor{PhamSvore2013}, the fan-out gate, can also be replaced by additional log-depth classical computations in the measurement-based quantum computing setting~\cite{Browne:2011}.

More recently, \citeauthor{Cirac:2021} introduced a scheme to implement unitary operations involving quantum circuits combined with Local Operations and Classical Communication ($\mathsf{LOCC}$) channels: $\mathsf{LOCC}$-assisted quantum circuits~\cite{Cirac:2021}. Similarly to the four-step scheme we just described, they allow for a short depth circuit to be run on the qubits, followed by one round of $\mathsf{LOCC}$, in which ancilla qubits are measured and local unitaries are applied based on the measurement outcomes. They show that in this model any 1D transitionally invariant matrix-product state (MPS) with fixed bond dimension is in the same phase of matter as the trivial state. Similar ideas can be found in~\cite{TVV_NonAbelianTopologicalOrder_2022, tantivasadakarn2021long}.

In this work, we introduce a new model, called \textit{Local Alternating Quantum-Classical Computations} ($\LAQCC$). In this model we alternate between running quantum circuits (constrained by locality), ending in the measurement of a subset of qubits, and fast classical computations based on the measurement results. The outcome of the classical computations are then used to control future quantum circuits. We allow for flexibility in this model, by giving different constraints to the power of both the quantum circuits and the classical circuits as well as the number of alternations between them. 
Most attention will be given to $\LAQCC$ containing quantum circuits of constant depth, classical circuits of logarithmic depth and at most a constant number of alternations between them. 
Any circuit constructed in this model is considered to be of constant depth. 
We restrict ourselves to logarithmic depth classical computations, as this is the first natural and non-trivial extension beyond constant-depth classical computations. 
Constant-depth classical computations do however also have an equivalent constant-depth quantum implementation.

The definition of $\LAQCC$ sharpens the original definition of \citeauthor{PhamSvore2013} by adding constraints to the intermediate classical computations. This allows us to bound the power of $\LAQCC$ from above. 

The main result of \citeauthor{Cirac:2021}, that 1D translational invariant MPS with fixed bond dimension can be prepared by $\mathsf{LOCC}$-assisted circuits, relies on local symmetries of the MPS. These symmetries allow them to prepare local states (on a constant number of qubits) and glue them together by doing one round of the appropriate entangling measurement and corrections, after which they run a round of local unitaries to get the desired result. This general scheme for preparing states that exhibit an MPS description with the appropriate local symmetries requires only geometrically local unitaries and one round of measurement and corrections an therefore is accessible in $\LAQCC$. Studying different local symmetries, known as Symmetry Protected Topological (SPT) phases of matter, to find measurement-based constant depth circuits for states is a broad ongoing field of research~\cite{TVV_NonAbelianTopologicalOrder_2022, tantivasadakarn2021long, smith2023deterministic}. 
All these schemes have a $\LAQCC$ implementation.


Note however that \citeauthor{Cirac:2021} also suggest a circuit for the $W$-state.
This circuit uses sequentially and dependent measurement-based corrections of the ancilla qubits. 
These dependent measurements translate to sequential alternations between the quantum and classical circuits and therefore increase the total depth to linear depth, exceeding the constant-depth constraints imposed by $\LAQCC$-circuits. 

We study the power of the $\LAQCC$ model with respect to state preparation, showing that even with only constant quantum-depth and logarithmic classical depth it remains possible to prepare states with long-range entanglement.
Another surprising result is that it is unlikely that $\LAQCC$ circuits are classically simulatable. We show that any instantaneous quantum polynomial-time (IQP) circuit~\cite{Bremner2010,Shepherd2009} has an $\LAQCC$ implementation.
Classical simulation of IQP circuits implies the collapse of the polynomial hierarchy to the third level, which is not believed to be true~\cite{Bremner2017}. Therefore, we expect that $\LAQCC$ circuits are unlikely to be classically simulatable. We bound the power of $\LAQCC$ by showing that it is contained in $\QNC^1$, the class of polynomial-size, log-depth circuits.

Next, we also study the power that intermediate classical calculations can add to quantum computations, by considering a new model that alternates between polynomially many polynomial-depth quantum circuits and unbounded classical computations
We study this model by doing a complexity theoretical analysis, where we draw inspiration from the notions of complexity given by \citeauthor{RosenthalYuen:2022}, \citeauthor{MetgerYuen:2023}, and \citeauthor{Aaronson:2004}.
All three complexity notions are based on the notion of state preparation, instead of more traditional definition of complexity such as the decidability of a computational problem. 
The first two consider classes based on sequences of quantum states preparable by a polynomial-sized quantum circuit, where the circuits are uniformly generated by a computational class, for instance, the class $\mathsf{PSPACE}$, which results in the complexity class $\mathsf{StatePSPACE}$~\cite{RosenthalYuen:2022,MetgerYuen:2023}.
The third notion considers a relative complexity, where the complexity is measured between two given states, and is measured by the number of gates, from a given gate-set, required to transform one state in another state~\cite{Aaronson:2004}. 
For our definition of state preparation complexity, we drop the uniformity constraint from~\cite{RosenthalYuen:2022,MetgerYuen:2023} and define a class as $\mathsf{StateX}$, which refers to states preparable by circuits of type $\mathsf{X}$. 
As an example, if $\mathsf{X} = \QNC^0$, this results in the class $\mathsf{StateQNC^0}$, which is the set of states preparable from the $\ket{0}^n$ state by poly-size constant-depth circuits. 
This notion is similar to the relative complexity from~\cite{Aaronson:2004}, where one state is the  $\ket{0}^n$ state and instead of counting the number of gates we consider the set of states preparable by a fixed number of gates. Using this notion of complexity we show that any state preparable by an $\LAQCC^*$ circuit is also preparable by a $\mathsf{PostQPoly}$ circuit, the class of circuits of polynomial depth with an additional post-selection gate. 

All Clifford circuits have a constant-depth $\LAQCC$ implementation, implying that any stabilizer state can be implemented by a constant-depth $\LAQCC$ circuit, see Section~\ref{sec:clifford_circuits} for a proof of this statement. 
Efficient circuits for stabilizer states have been known already through measurement-based quantum computing. Therefore this paper focuses on the preparation of non-stabilizer states, and as a surprising result we find novel constant-depth protocols for four very natural classes of non-stabilizer states.
Despite the extensive research into these four classes of non-stabilizer states and the many applications of them, no efficient constant- or low-depth state preparation protocols are known yet. We specifically consider these four classes as they are all often used as initial states in other algorithms.

The first state is a uniform superposition over an arbitrary number of states. 
This state finds applications in many quantum algorithms, as they often start with a uniform superposition over multiple states. 
This superposition is often achieved by applying Hadamard gates to every qubit due to its simplicity to prepare. 
Yet, the analysis of many algorithms, such as Shor's algorithm~\cite{Shor:1997}, would benefit from a different initial superposition. 
The circuit to prepare the uniform superposition over an arbitrary number of states uses an exact version of Grover search as a subroutine, that turns a probabilistic circuit, with a known constant probability of success, into a deterministic circuit. 
We use the circuit for preparing a uniform superposition over an arbitrary number of states as a subroutine in the next two quantum state preparation protocols. 

The second state is the $W$-state, the uniform superposition over all computational basis states of Hamming-weight~$1$, a natural long-ranged entangled state that displays a fundamentally nonequivalent type of entanglement from the Greenberger–Horne–Zeilinger state~\cite{WState:2000}, for which $\LAQCC$-type constant-depth circuits were previously known~\cite{PhamSvore2013, Cirac:2021}. 
The $W$-state is often used as benchmark for new quantum hardware~\cite{Haffner2005,Neeley2010,GarciaPerez:2021}. 
A novel way to prepare the $W$-state therefore gives a new way to benchmark different quantum devices with each other. 
A circuit for preparing the $W$-state was given in~\cite{Cirac:2021}, but this implementation requires sequentially alternating measurements followed by local unitaries, which in the $\LAQCC$ model is not considered to be of constant depth. 
We improve this protocol by giving an $\LAQCC$ implementation of the $W$-state, based on a compress-uncompress method that links the one-hot and binary encoding of integers.

The third state considered is the Dicke state, a generalization of the $W$-state, a superposition over all computational basis states with Hamming-weight $k$~\cite{Dicke:1954}. 
Dicke states have relevance in various practical settings.
For instance, for quantum game theory~\cite{zdemir2007}, quantum storage~\cite{Bacon_Compress:2006,Plesch:2010}, quantum error correction~\cite{ouyang2014permutation}, quantum metrology~\cite{toth2012multipartite}, and quantum networking~\cite{prevedel2009experimental}. 
Dicke states have been used as a starting state for variational optimization algorithms, most notably Quantum Alternating Operator Ansatz (QAOA)~\cite{Hadfield2019}, to find solutions to problems such as Maximum k-vertex Cover~\cite{Brandhofer2022,cook2020quantum}.
The ground states of physical Hamiltonians describing one-dimensional chains tend to show a resemblance to Dicke states such as states resulting from the Bethe ansatz, making them an ideal starting state when investigating the ground state behavior of these Hamiltonians~\cite{TDL_BetheAnsatzDerivation:2010,B_ExcitedStateQuantumPhaseTransitions:2013,DickeTransitions:2021}. 
For instance, the algorithm by \citeauthor{van2021preparing}, who give an algorithm to prepare the Bethe ansatz eigenstates of the spin-1/2 XXZ spin chain, starts by first preparing a Dicke state~\cite{van2021preparing}. 
A Dicke-state preparation protocol based on the compress-uncompress methodology used in the $W$-state furthermore finds applications in entanglement distillation, where the entanglement of a large state is concentrated on only a few qubits. 
Efficient deterministic circuits for preparing Dicke states have been proposed by \citeauthor{bartschi2019deterministic}~\cite{bartschi2019deterministic, bartschi2022deterministic_short_depth}. 
They provide a quantum circuit of depth $\mathO(k \log(\frac{n}{k}))$, allowing arbitrary connectivity, to prepare a Dicke state, which they conjecture to be optimal when $k$ is constant. 
In this work, we provide a constant-depth $\LAQCC$ circuit below their conjectured bound already for constant $k$. 
However, this does not directly disprove their conjecture, as we allow for intermediate measurements and classical computations. 
More significantly, we even construct constant-depth $\LAQCC$ circuits for $k = \mathO(\sqrt{n})$ greatly improving their bound.
This construction extends the compress-uncompress method for the $W$-state combined with additional subroutines. 

We continue with a log-depth state preparation protocol for the Dicke-state for arbitrary $k$. 
This protocol implements an efficient transformation between the factoradic number representation and the combinatorial number representation of a positive integer. 
The combinatorial number representation relates directly to the Dicke state. 
The provided efficient transformation between number representation systems might be of independent interest. 

We conclude by modifying our protocol for preparing a Dicke-state to a protocol that prepares quantum many-body scar states in constant-depth. 
These states have low entanglement and longer coherence times than states with similar energy density.
These characteristics make many-body scar states interesting to analyze and relevant within physics.
Many-body scar states appear for instance in the AKLT model~\cite{AKLT:1987,MRBAR:2018,MRB:2018} and different spin models~\cite{SI:2019,MOBFR:2020}.
Known methods for preparing these states have polynomial-depth~\cite{Gustafson:2023}, whereas our circuit has constant depth.

\paragraph{Summary of results}
\begin{itemize}
    \item We give a new definition of a computational model that captures the power of the four step process: applying a constant number of layers of one- and two-qubit gates; performing a syndrome measurement; perform a fast classical computation determining corrections; apply corrections. We call this model \emph{Local Alternating Quantum Classical Computations}, or $\LAQCC$ for short. In this model we bound the allowed quantum operations, intermediate classical calculations, and number of rounds separately. In Section~\ref{sec:LAQCC_model} we define this model and give a list of operations based on results from literature contained in this computational model. In some of these operations we explicitly use that we allow for multiple, but at most constant, rounds  of corrections.
    \item  We show show that there exist $\LAQCC$ circuits that can not be weakly simulated in Section~\ref{sec:IQP_in_LAQCC}. We further show that for every $\LAQCC$ circuit there exists a $\QNC^1$ circuit simulating it perfectly, in Section~\ref{sec:LAQCC_in_QNC1}.
    \item We introduce a new type computational complexity for preparing states and show that the extension of $\LAQCC$ where we allow a polynomial number of rounds and unbounded classical computation, is contained in $\mathsf{PostQPoly}$, the class of polynomial circuits with post-selection, in Section~\ref{sec:Complexity results}.
    \item We show a protocol to prepare the uniform superposition state of size $q$ in $\LAQCC$ using $\mathO(\ceil{\log_2(q)}^2)$ qubits in Section~\ref{sec:superposition_modulo_q}. 
    \item We show a protocol to prepare the $W_n$ state in $\LAQCC$ using $\mathO(n\log(n))$ qubits in Section~\ref{sec:W_state_in_LAQCC}.
    \item We show two ways of preparing the Dicke-$(n,k)$ state. The first method is in $\LAQCC$, works up to $k = \mathO(\sqrt{n})$, uses $\mathO(n^2\log(n))$ qubits, and is found in Section~\ref{sec:dicke:small_k}. The second method is in $\LAQCC\text{-}\mathsf{LOG}$ (an extension of $\LAQCC$ allowing for logarithmic number of alterations instead of constant), works for any $k$, uses $\mathO(\text{poly}(n))$ qubits, and is found in Section~\ref{sec:Dicke_in_LAQCC_LOG}. 
    \item We extend on our $\LAQCC$ method of generating Dicke-$(n,k)$ states for $k = \mathO(\sqrt{n})$ and show a protocol to generate many-body scar states for a particular Hamiltonian in $\LAQCC$ (Section~\ref{sec:many_body_scar}). 
\end{itemize}
Summarized in a table, we provide the following state generation protocols:
\begin{table}[htb]
\centering
\begin{tabular}{l|l|l|l}
\textbf{State description} & \textbf{Width} & \textbf{Depth} & \textbf{Implementation}\\
\hline 
Uniform superposition mod $q$: $\frac{1}{\sqrt{q}} \sum_{i = 0}^{q-1}\ket{i}$ & $\mathO(\ceil{\log^2 q})$ & $\mathO(1)$ & Section~\ref{sec:superposition_modulo_q}\\

$W$-state: $\frac{1}{\sqrt{n}}\sum_{i = 0}^{n-1}\ket{e_i}$ & $\mathO(n \log n)$ & $\mathO(1)$ & Section~\ref{sec:W_state_in_LAQCC}\\

Dicke-$(n,k)$, $k = \mathO(\sqrt{n})$: $\binom{n}{k}^{-1/2}\sum_{x \in \{0,1\}^n: |x| = k} \ket{x}$ &  $\mathO(n^2\log n)$ & $\mathO(1)$ 
&Section~\ref{sec:dicke:small_k}\\

Dicke-$(n,k)$: $\binom{n}{k}^{-1/2}\sum_{x \in \{0,1\}^n: |x| = k} \ket{x}$ & $\mathO(\text{poly}(n))$ & $\mathO(\log n)$ &Section~\ref{sec:Dicke_in_LAQCC_LOG}\\

QMBS: $\ket{S_k} = \frac{1}{k! \sqrt{\mathcal N(n,k)}}(Q^\dagger)^k \ket{\Omega}$ &  $\mathO(n^2\log n)$ & $\mathO(1)$  &  Section~\ref{sec:many_body_scar}
\end{tabular}
\caption{Summary of state preparation protocols given in this paper.}
\label{tab:sate_prep}
\end{table}
In the entry for the quantum many-body scar state $Q$ denotes the raising operator and $\mathcal N(n,k)=\binom{n-k-1}{k}$. 
Section~\ref{sec:many_body_scar} will provide more details on the variables and the implementation. 

\paragraph{Organization of the paper}
\noindent We first introduce relevant preliminaries in Section~\ref{sec:preliminaries}. 
In Section~\ref{sec:LAQCC_model} we formally define the class of Local Alternating Quantum-Classical Computations ($\LAQCC$). We also show that any Clifford circuit can be implemented in constant depth $\LAQCC$ (a result based on a result from measurement-based quantum computing~\cite{jozsa2006introduction}). 
This result allows us to give many useful multi-qubit gates and routines in Section~\ref{sec:gates_created_in_LAQCC}. 
Beyond that we show that constant depth $\LAQCC$ circuits are contained in $\QNC^1$ and that any $\mathsf{IQP}$ circuit has an $\LAQCC$ implementation.
We conclude this section with an analysis of a more powerful instantiation of $\LAQCC$ and show an inclusion with respect to the class $\mathsf{PostQPoly}$, which is the class of circuits of polynomial depth with one additional post-selection gate. 
In Section~\ref{sec:state_prep_in_LAQCC} we give $\LAQCC$ circuit implementations for preparing the uniform superposition over an arbitrary number of states, the $W$-state and the Dicke state up to $k = \mathO(\sqrt{n})$. We furthermore give a log-depth circuit implementation for preparing the Dicke state for any $k$. We conclude by showing a $\LAQCC$ circuit for generating many body scar states of a particular type of Hamiltonian.

%% file: preliminaries.tex
\section{Preliminaries}
\label{sec:preliminaries}
In this section we recap definitions used throughout the rest of the paper.

\subsection{Complexity classes}
The computational model introduced in this work uses complexity classes. 
Typically these classes are defined as classes of decision problems solvable by some type of circuits. 
Below we give definitions of some of these complexity classes in terms of the circuits contained in that class. 

\begin{definition}
The class $\NC^k$ consists of all decision problems solvable by circuits of $\mathO((\log n)^k)$ depth and polynomial size, and consisting of bounded-fan-in AND- and OR-gates.

The class $\AC^k$ consists of all decision problems solvable by circuits of $\mathO((\log n)^k)$ depth and polynomial size, and consisting of unbounded-fan-in AND- and OR-gates.

The class $\TC^k$ consists of all decision problems solvable by circuits of $\mathO((\log n)^k)$ depth and polynomial size, and consisting of unbounded-fan-in AND-, OR- and Threshold$_t$-gates, where a Threshold$_t$-gate evaluates to one if and only if the sum of the inputs is at least $t$. 
\end{definition}
These classes also have a quantum equivalent class. 
\begin{definition}
The class $\QNC^k$ consists of all decision problems solvable by quantum circuits of $\mathO((\log n)^k)$ depth and polynomial size, and consisting of single- and two-qubit quantum gates.
\end{definition}
Definitions for the quantum versions of $\AC^k$ and $\TC^k$ also exist.
However, when equipping the class $\QNC^k$ with unbounded-fan-in parity gates, all three classes intersect~\cite{GreenHomerMoorePollett:2002,Moore:1999,TakahashiTani_CollapseOfHierarchyConstantDepthExactQuantumCircuits_2013}.

Two other often used classes are $\mathsf{P}$ and $\mathsf{L}$.
\begin{definition}
The class $\mathsf{P}$ consists of all decision problems solvable in polynomial time by a Turing machine. 

The class $\mathsf{L}$ consists of all decision problems solvable using only a logarithmic amount of memory. 
\end{definition}
The first class poses a limit on the depth of the operations performed by the Turing machine. 
The second class instead limits the available memory. 
Note however that with a logarithmic amount of memory, only a polynomial number of states is available, and hence the computation time is polynomial as well. 

Relations between different complexity classes exists: for instance, for all $k$, $\NC^k\subseteq\AC^k\subseteq\TC^k$. 
By \citeauthor{Johnson:1990}, we also have the inclusion $\mathsf{L} \subseteq \AC^1 \subseteq \TC^1$~\cite{Johnson:1990}.

In the remainder of this work we abuse notation and refer to $\mathsf{X}$-circuits as circuits that correspond to a decision problem in the class $\mathsf{X}$.
For example, an $\NC^k$ circuit is a circuit of $\mathO((\log n)^k)$ depth and polynomial size that corresponds to a decision problem in $\NC^k$.

Finally, we define the class of polynomial sized quantum circuits:~\footnote{Due to the bounded-error-aspect associated to decision problems in $\BQP$, we follow this definition instead of talking about $\BQP$-circuits.}
\begin{definition}
The class $\mathsf{QPoly}(n)$ consists of all polynomial-sized quantum circuits (in $n$) that use single and two-qubit quantum gates.
\end{definition}

\subsection{Quantum gate sets}
First, recall the definition of the Pauli group and the Clifford group. 
\begin{definition}
The one qubit Pauli-group $P_1$ consists of four matrices, the identity matrix $I$ and the three Pauli matrices: 
$$X = \begin{pmatrix}0&1\\1&0\end{pmatrix}, \qquad Y = \begin{pmatrix}0&-i\\i&0\end{pmatrix}, \qquad Z = \begin{pmatrix}1&0\\0&-1\end{pmatrix},$$
together with a global phase of $\pm 1$ or $\pm i$. 

The $n$-qubit Pauli-group $P_n$ is the set of all $4^{n+1}$ possible tensors of length $n$ of matrices from $P_1$, together with a global phase of $\pm 1$ or $\pm i$. 
\end{definition}
The Clifford group forms the other well-known group of quantum circuits, as it stabilizes the Pauli group. 
\begin{definition}
The Clifford group $C_n$ consists of all $n$-qubit unitaries that leave the Pauli group 
$P_n$ invariant under conjugation. 
That is, let $c\in C_n$ be any Clifford circuit, then for any $P\in P_n$, there exists a $P'\in P_n$, such that $cP = P'c$.

The Clifford group is generated by the $CNOT$-gate, the Hadamard gate $H$ and the phase gate $S$, that act on computational basis states $\ket{x}$ and $\ket{y}$ via
$$CNOT: \ket{x}\ket{y}\mapsto \ket{x}\ket{y\oplus x}, \qquad H: \ket{x} \mapsto \ket{0}+(-1)^x \ket{1}, \qquad S: \ket{x} \mapsto i^x \ket{x}.$$
Any circuit constructed using only these three gates is called a Clifford circuit. 
\end{definition}

Unsurprisingly, Clifford circuits only cover a small part of the possible quantum circuits. 
Moreover, on a linear nearest-neighbor architecture, $\mathO(n)$ deep Clifford circuits suffice to simulate any Clifford unitary of size $2^n\times 2^n$~\cite{maslov2018shorter}.
We can furthermore simulate Clifford circuits efficiently~\cite{GottesmanKnill:1998}.
Universal quantum computations require additional gates, though almost any quantum gate suffices. 
For example, adding the single qubit $T$-gate, $T:\ket{x}\mapsto e^{i\pi x/4}\ket{x}$, to the Clifford group gives a universal gate set. 

\subsection{Two quantum subroutines}
This section discusses two quantum subroutines used in later sections. 
The first concerns Grover's algorithm with zero failure probability~\cite{Long_GroverZeroFailureRate_2001}. 
The second concerns parallelization of commuting gates using quantum fan-out gates~\cite{HoyerSpalek:2005}.

Grover's search algorithm gives a quadratic speed-up for unstructured search~\cite{Grover:1996}. 
After sufficient iterations, a measurement returns a target state with high probability. 
Surprisingly, if the exact number of target states is known, a slight modification of the Grover iterates allows for returning a target state with certainty, assuming noiseless computations. 
Lemma~\ref{lem:grover_constant_fraction} uses the next lemma to prepare quantum states instead of to find a target state.
\begin{lemma}[\cite{Long_GroverZeroFailureRate_2001}]\label{lem:exact_grover}
Let $L$ be a set of items and $T\subseteq L$ a set of targets, with $N=|L|$ and $m=|T|$ both known. 
Let $g : L \rightarrow \{0,1\}$ label the items in $L$ and define the oracle $O_g: \ket{x}\ket{b} \mapsto \ket{x}\ket{b\oplus g(x)}$.

Then, there exists a quantum amplitude amplification algorithm that makes $\mathO(\sqrt{N/m})$ queries to $O_g$ and prepares the quantum state $\frac{1}{\sqrt{|T|}}\sum_{x\in T}\ket{x}$.
\end{lemma}

For the other result, we use the quantum fan-out gate, that implements the map 
$$\ket{x}\ket{y_1}\hdots\ket{y_n}\mapsto\ket{x}\ket{y_1\oplus x}\hdots\ket{y_n\oplus x}.$$
Section~\ref{sec:gates_created_in_LAQCC} gives more details on how to implement this gate. 
\citeauthor{HoyerSpalek:2005} introduced this gate and analyzed its properties. 
The state preparation protocols given in Section~\ref{sec:state_prep_in_LAQCC} use the property that the quantum fan-out gate allows parallelization of commuting quantum gates~\cite{HoyerSpalek:2005}. 
\begin{lemma}\emph{(\cite[Theorem~3.2]{HoyerSpalek:2005})}
\label{lem:unitar_parallelization}
Let $\{U_i\}_{i=1}^n$ be a pairwise commuting set of gates on $k$ qubits. 
Let $U_i^{x_i}$ be the gate $U_i$ controlled by qubit $\ket{x_i}$. 
Let $T$ be the unitary that mutually diagonalizes all $U_i$. 
Then there exists a quantum circuit, using quantum fan-out gates, computing $U = \prod_{i=1}^n U_i^{x_i}$ with depth $\max_{i = 1}^n \mathrm{depth}(U_i) + 4\cdot\mathrm{depth}(T) + 2$ and size $\sum_{i = 1}^n \mathrm{size}(U_i) + (2n + 2) \cdot \mathrm{size}(T)+ 2n$, using $(n - 1)k$ ancilla qubits.
\end{lemma}

%% file: model.tex
\section{The \texorpdfstring{$\LAQCC$}{LAQCC} model}
\label{sec:LAQCC_model}
This section formally defines the $\LAQCC$ model.
Next, we show that all Clifford circuits have an efficient and equivalent $\LAQCC$ circuit.
We then give quantum gates and tools constructable within $\LAQCC$, such as the quantum fan-out gate and weighted threshold gate, and we conclude by showing that any $\LAQCC$-circuit has an equivalent $\QNC^1$-circuit.

\subsection{Model definition}
We define the computational model \emph{Local Alternating Quantum-Classical Computations} ($\LAQCC$) as follows:
\begin{definition}[Local Alternating Quantum-Classical computations]
Let $\LAQCC(\mathcal{Q}, \mathcal{C}, d)$ be the class of circuits such that
\begin{itemize}
\item every quantum layer implements a quantum circuit $Q\in\mathcal{Q}$ constrained to a grid topology;
\item every classical layer implements a classical circuit $C\in\mathcal{C}$;
\item there are $d$ alternating layers of quantum and classical circuits;
\item after every quantum circuit $Q$ a subset of the qubits is measured;
\item the classical circuit receives input from the measurement outcomes of previous quantum layers;
\item the classical circuit can control quantum operations in future layers.
\end{itemize}
The allowed gates in the quantum and classical layers are given by $\mathcal{Q}$ and $\mathcal{C}$ respectively. 
Furthermore, we require a circuit in $\LAQCC(\mathcal{Q}, \mathcal{C}, d)$ to deterministically prepare a pure state on the all-zeroes initial state. 
\end{definition}

The grid topology imposed on the quantum operations implies that qubits can only interact with their direct neighbors on the grid. 
A circuit in $\LAQCC(\mathcal{Q}, \mathcal{C}, d)$ can use the results of the classical intermediate layers and control quantum operations in future layers. 
In a sense, information is fed forward in the circuit. Note, as classical computations are in general significantly faster than the quantum operations, we only count the quantum operations towards the depth of the circuit, unless specified otherwise.

\begin{remark}
Note that there exists ambiguity in the choices for $\mathcal{Q}$ and $\mathcal{C}$.
For example, we have $\LAQCC(\mathsf{QPoly}(n), \mathsf{P}, \mathO(1)) = \LAQCC(\QNC^0, \mathsf{P}, \mathO(\mathrm{poly}(n)))$. 
This follows as any $\mathsf{P}$-circuit is in $\mathsf{QPoly}(n)$, and we can concatenate $\mathrm{poly}(n)$ constant-depth quantum circuits with trivial intermediate classical computations. 

This ambiguity is non-trivial: consider for instance 
$$\LAQCC(\QNC^1, \NC^1, \mathO(1))\subseteq \LAQCC(\QNC^0, \NC^1, \mathO(\log(n))).$$ 
The inclusion from left to right follows immediately by same argument as above.
It is however not obvious if the logarithmic number of measurement rounds, allowed in the right hand side, can be simulated by a $\QNC^1$ circuit.
Even stronger, we will show in Section~\ref{sec:gates_created_in_LAQCC} that threshold gates are available in $\LAQCC(\QNC^0, \NC^1, \mathO(1))$. From this fact it follows immediately that any $\TC^1$-circuit is contained in $\LAQCC(\QNC^0, \NC^1, \mathO(\log(n)))$. 
It is unclear these circuits are also contained in $\LAQCC(\QNC^1, \NC^1, \mathO(1))$. 
\end{remark}

In the remainder of this work, we consider a specific instantiation of $\LAQCC(\mathcal{Q},\mathcal{C}, d)$.
\begin{notation}
We let $\LAQCC$ refer to the instance $\LAQCC(\QNC^0, \NC^1, \mathO(1))$, together with a grid nearest-neighbor topology and a quantum gate-set of all single-qubit gates and the two-qubit CNOT gate. 
The classical computations are bounded to logarithmic depth and of bounded-fan-in. 
\end{notation}
The class $\NC^1$ is a natural non-trivial class beyond constant-depth complexity classes.
As the depth of these circuits remains low, they can be implemented quickly. 
In this work, we assume qubits will not decohere during the classical computation.
Incorporating the errors throughout the whole $\LAQCC$-computation, including the classical intermediate computations, proves an interesting direction for future research.
Interesting, $\LAQCC$ contains many useful gates and subroutines already, as outlined in the next section, specifically the fan-out gate, implying directly that threshold gates are also contained in $\LAQCC$.

In its current definition, $\LAQCC(\mathcal{Q},\mathcal{C}, d)$, and hence also $\LAQCC$, are classes of circuits. 
When considering the capabilities of $\LAQCC(\mathcal{Q},\mathcal{C}, d)$ in preparing states, it is helpful to define a related class that consists of states preparable by a circuit in $\LAQCC(\mathcal{Q},\mathcal{C}, d)$. 
\begin{restatable}{definition}{StateClassX}
Let $\mathcal{H}_n$ be a Hilbert space on $n$ qubits, then define 
$$\mathsf{StateX}_{n,\varepsilon}=\{\ket{\psi}\in\mathcal{H}_n\mid \exists\, \mathsf{X}\text{-circuit }A:
\bra{\psi}A\ket{0}^{\otimes n} \ge 1-\varepsilon\}.$$

This is the subset of $n$-qubit states $\ket{\psi}$ such that there exists a circuit corresponding to the class~$\mathsf{X}$ that prepares a quantum state that has inner product at least $1-\varepsilon$ with $\ket{\psi}$. 

Define $\mathsf{StateX}_{\varepsilon} =\bigcup_{n\in\mathbb{N}}\mathsf{StateX}_{n,\varepsilon}$.
\end{restatable}
This definition extends already existing ideas and definitions of state-complexity~\cite{aaronson2020hardness,RosenthalYuen:2022,susskind2018lectures}. Our definition is very similar to state complexity defined in \cite{MetgerYuen:2023}, where we are interested in which states are contained in a class, however we drop the uniformity requirement and instead study the set of states that can be generated by a specific class of circuits.
An example of a circuit class is $\mathsf{State}$$\LAQCC(\mathcal{Q},\mathcal{C}, d)_{n,\varepsilon}$.
\begin{notation}
The class $\mathsf{State}$$\LAQCC(\mathcal{Q},\mathcal{C}, d)_{n,\varepsilon}$ consists of all $n$-qubit states $\ket{\psi}$ for which an $\LAQCC(\mathcal{Q},\mathcal{C}, d)$ exists that prepares a state that has inner product at least $1-\varepsilon$ with $\ket{\psi}$. 
\end{notation}

Another example is the circuit class of $\mathsf{PostQPoly}$. 
\begin{restatable}{definition}{PostQPolyDef}
\label{def:PostQPoly}
The class $\mathsf{PostQPoly}$ consists of all polynomial-sized quantum circuits with one extra qubit, where the outcome state is considered conditional on the extra qubit being in the one state. 
If the extra qubit is in the zero state, the output state may be anything. 

The class $\mathsf{StatePostQPoly}_{n,\varepsilon}$ consists of all $n$-qubit states $\ket{\psi}$ for which a polynomial-sized quantum circuit exists that prepares a state that, conditional on the extra qubit being one, has inner product at least $1-\varepsilon$ with $\ket{\psi}$. 
\end{restatable}

The next section shows that all polynomial-size Clifford circuits are in $\LAQCC$. 
As a result, the quantum fan-out gate is also in $\LAQCC$, a multi-qubit gate which enables many different subroutines (by applying Lemma~\ref{lem:unitar_parallelization}).
We give a list of such quantum gates and quantum subroutines accessible $\LAQCC$ in Section~\ref{sec:gates_created_in_LAQCC}. 
We conclude by showing that any $\LAQCC$-circuit corresponds to a $\QNC^1$-circuit.

\subsection{Clifford circuits}
\label{sec:clifford_circuits}
The concept of intermediate measurements with subsequent computations is closely related to measurement-based quantum computing. 
A famous result from measurement-based quantum computing us that all Clifford circuits can be paralellized using measurements. In this section we borrow techniques from this result to show that any Clifford circuit has an $\LAQCC$ implementation. 

This result is best understood in the teleportation based quantum computing model~\cite{jozsa2006introduction}, a specific instance of measurement-based quantum computing that applies quantum operations using bell measurements. 
In teleportation, qubits are measured in the Bell basis, which projects the measured qubits onto an entangled two-qubit, or ebit, state, up to local Pauli gates. This projection combined with an ebit state teleports a quantum state between qubits. After teleportation, one needs to correct the local Pauli gate created by the bell measurement. A similar process can be used to apply quantum gates. However, the Pauli gates that arise during teleportation have to be corrected before the calculations can proceed, which necessitates subsequent adaptive operations.

With Clifford circuits, these subsequent operations can be omitted. 
Clifford circuits stabilize the Pauli group, which allows for simultaneous measurements and hence parallelization of the full Clifford circuit~\cite{jozsa2006introduction}. 
Consider a simple example of teleporting a single-qubit quantum state. 
A Bell basis measurement projects two qubits on $\sum_{i\in \{0,1\}} \bra{ii}P^{a,b}\otimes I$, where $P^{a,b} = Z^a X^b$ and $a,b \in \{0,1\}$ correspond to the four possible measurement outcomes.

By using one Bell-basis measurement, we can apply two sequential Clifford gates $U_1$ and $U_2$ on a quantum state $\ket{\psi}$, which gives:
\begin{align*}
    \sum_{i, j\in \{0,1\}} \big[ (\bra{ii} (P^{a,b} \otimes I) \otimes I \big] U_1 \otimes I \otimes U_2 \ket{\psi}\ket{jj} &= \sum_{i,j \in \{0,1\}} \bra{i} P^{a,b} U_1 \ket{\psi} \braket{i}{j} U_2\ket{j}\\
    &=  \sum_{i \in \{0,1\}}  U_2\kb{i} P^{a,b} U_1 \ket{\psi} = U_2 P^{a,b} U_1\ket{\psi}.
\end{align*}
Note that besides projecting on a Bell state, an initial entangled Bell-state is required. 
$U_2$ is a Clifford gate, hence there exists a $P^{\hat{a},\hat{b}}$ such that $U_2 P^{a,b} U_1\ket{\psi} = P^{\hat{a},\hat{b}} U_2  U_1\ket{\psi}$, allowing the correction term to be pushed to the end of the circuit.
Repeating the same argument for multiple Clifford unitaries gives the quantum state $...P^{a_2,b_2}_2 U_2P^{a_1,b_1}_1 U_1\ket{\psi}$. 
Due to the conjugation relation of Clifford and Pauli gates, all correction terms can be postponed to the end of the computation. 

\subsubsection{Clifford-ladder circuit}
A similar argument holds when looking at \textit{Clifford-ladder circuits}. 
\begin{definition}[Clifford-ladder circuit]
Let $\{U^{i}\}_{i=0}^{n-1}$ be a collection of $n$ $2$-qubit Clifford unitaries. 
A Clifford-ladder circuit $C_{ladder}$ is a circuit of depth $\mathO(n)$ and width $\mathO(n)$ of the following form:
\[
C_{ladder} = \prod_{i=0}^{n-1} U^{(i)}_{i,i+1}
\]
where $U^{(i)}_{i,i+1}$ denotes that unitary $U^{(i)}$ is applied on qubits $i$ and $i+1$.
\end{definition}
Note that each 2-qubit Clifford unitary $U^{(i)}$ itself is of constant-depth. 

The next lemma shows that any Clifford-ladder circuit has an equivalent $\LAQCC$ circuit. 
Figure~\ref{fig:clifford_ladder} shows this mapping graphically.
Each two-qubit unitary is parallelized using gate teleportation and with the Clifford commutation relations, the Pauli correction terms are pushed to the end of the computation. 
\begin{lemma}
\label{lem:clifford_ladder}
Any Clifford-ladder circuit has an $\LAQCC$ implementation of depth $\mathO(1)$ and width $\mathO(n)$.
\end{lemma}
\begin{proof}
Figure~\ref{fig:clifford_ladder} shows the construction of a $\LAQCC$ circuit of width $\mathO(n)$ and depth $\mathO(1)$ implementing a Clifford-ladder circuit.
The caps and cups denote Bell-state measurements and Bell-state creation, respectively. 
What remains to show is that an $\NC^1$ circuit computes the Pauli-correction terms.

The $i$-th Bell measurement results in Pauli error $P_i = Z^{a_i}X^{b_i}$. 
A Clifford-ladder circuit of size $n$ hence has an error vector $\big(a\,b\big)$ of length $2n$. 
The correction terms that have to be applied have the same form: we can label every corrective Pauli by an index $j$, such that $\hat{P}_j = Z^{\hat{a}_j}X^{\hat{b}_j}$. 
This gives a correction vector $\big(\hat{a}\,\hat{b}\big)$. 
Note that Pauli matrices anti-commute, hence reordering them will only incur a global phase.
This implies a binary linear map $M:\big(a\,b\big)\mapsto\big(\hat{a}\,\hat{b}\big)$. 
As matrix vector multiplication is in $\NC^1$, this error calculation is in $\NC^1$ and Clifford-ladder circuits have an $\LAQCC$ implementation.
\end{proof}
\begin{figure}
    \centering
    \includegraphics[width = \textwidth]{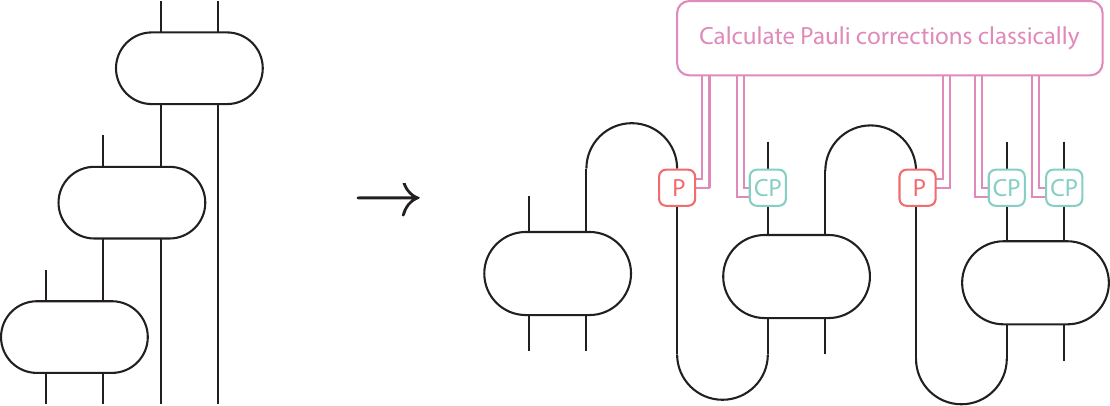}
    \caption{Graphical representation of Clifford-ladder circuit parallelization. 
    Time flows upward and lines represent qubits and boxes quantum gates. A half circle represents either a Bell-state creation (ends pointing upwards) or a Bell-state measurement (ends pointing downwards). The Bell-state measurements can produce Pauli errors $P = Z^a X^b$, which are corrected by the boxes $CP$ (corrective Pauli). 
    The computations to determine how errors propagate are performed classically before the computations.}
    \label{fig:clifford_ladder}
\end{figure}

\begin{remark}
Constructing the binairy linear map $M$ is not in $\NC^1$, but it does follow directly from the quantum circuit.
Instead, an $\mathsf{L}$ (logspace) precomputation gives the matrix associated to $M$.
\end{remark}

This result directly implies that in $\LAQCC$ we can apply two-qubit gate on any two any non-adjacent qubits. 
\begin{corollary}
Any $SWAP$ circuit needed to do an operation between non-adjacent qubits is a Clifford-ladder circuit and hence in $\LAQCC$.

This effectively removes the locality constraint in $\LAQCC$ for applying a single $2$-qubit gate on non-adjacent qubits. 
\end{corollary}

An example of a Clifford-ladder circuit is the creation of a GHZ state. 
We can parallelize this directly, for instance following the poor man's cat state approach of~\cite{WattsKothariSchaefferTal:2019}. 
Figure~\ref{fig:q_circuit:GHZ_3_state} shows a $\LAQCC$ circuit using $2n-1$ qubits placed on a line that prepares an $n$-qubit GHZ state.
\begin{figure}
\centering
\includegraphics[width=0.8\textwidth]{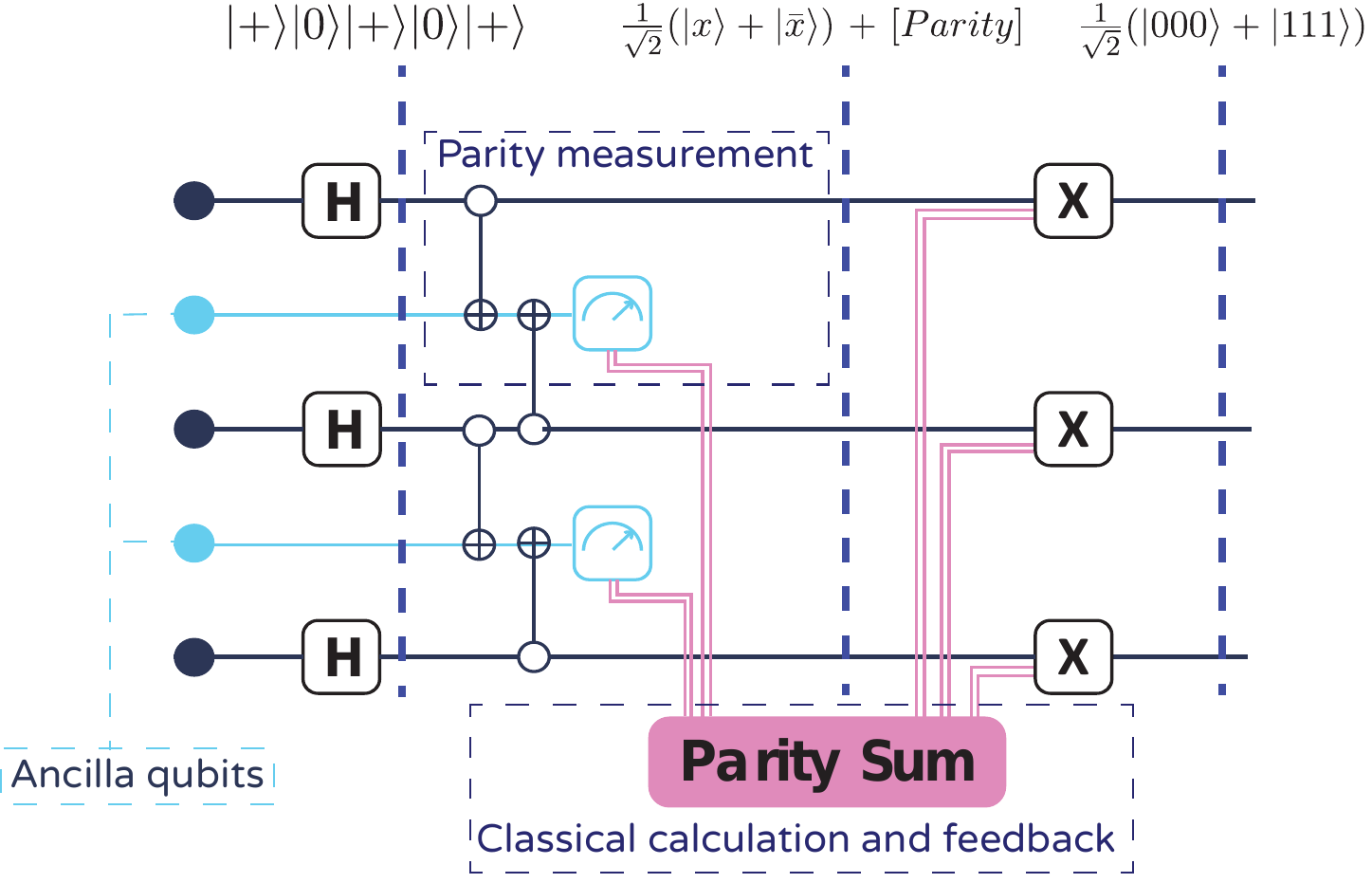}
\caption{The quantum circuit to prepare the $3$-qubit GHZ state. 
Double lines indicate classical information and dotted lines the quantum state at various points.}
\label{fig:q_circuit:GHZ_3_state}
\end{figure}

\subsubsection{Clifford-grid circuit}
Any Clifford unitary can be mapped to a linear-depth circuit given a linear nearest-neighbor architecture~\cite{maslov2018shorter}. 
The most general representation of these circuits are so-called Clifford-grid circuits.
\begin{definition}[Clifford-grid circuit]
Let $n$ be the number of qubits. 
A Clifford-grid circuit of depth $d$ is a circuit of the form
\[
   C_{grid} =  \prod_{i=0}^{d}\bigotimes_{j=0}^{\frac{n}{2}}U_{i,j},
\]
for Clifford unitaries $U_{i,j}$ and such that gate $U_{i,j}$ acts on qubits $2j$ and $2j + 1$ if $i$ is even, and $2j+1$ and $2j+2$ is $i$ is odd. 
\end{definition}

The next lemma shows that Clifford-grid circuits also have an efficient $\LAQCC$ implementation. 
\begin{lemma}
Any Clifford-grid circuit of depth $\mathO(n)$ has an $\LAQCC$ implementation of depth $\mathO(1)$ and width $\mathO(n^2)$.
\label{lem:clifford_grid}
\end{lemma}
\begin{proof}
Similar to the Clifford-ladder circuits, gate teleportation allows parallelization to obtain a $\LAQCC$ circuit. 
With a total of $\mathO(n^2)$ Clifford gates, this also requires $\mathO(n^2)$ qubits. 
Figure~\ref{fig:clifford_grid} illustrates the transformation. 

Any Bell measurement in the circuit can incur a Pauli error, which has to be dealt with at the end of the circuit. 
The number of Pauli gates now scales with $\mathO(n^2)$. 
Similar to the Clifford-ladder circuits, there now is a vector $(a b)$ of length $\mathO(n^2)$ containing the information of the Pauli errors. 
The vector of correction terms, the vector $(\hat{a} \hat{b})$, has length $\mathO(n)$.

As these Pauli errors anti-commute, there again is a binary linear map $M: (a b)\mapsto (\hat{a}\hat{b})$. 
The corresponding matrix is rectangular and the error-correction calculations are in $\NC^1$.
\end{proof}
\begin{figure}[ht]
    \centering
    \includegraphics[width = \textwidth]{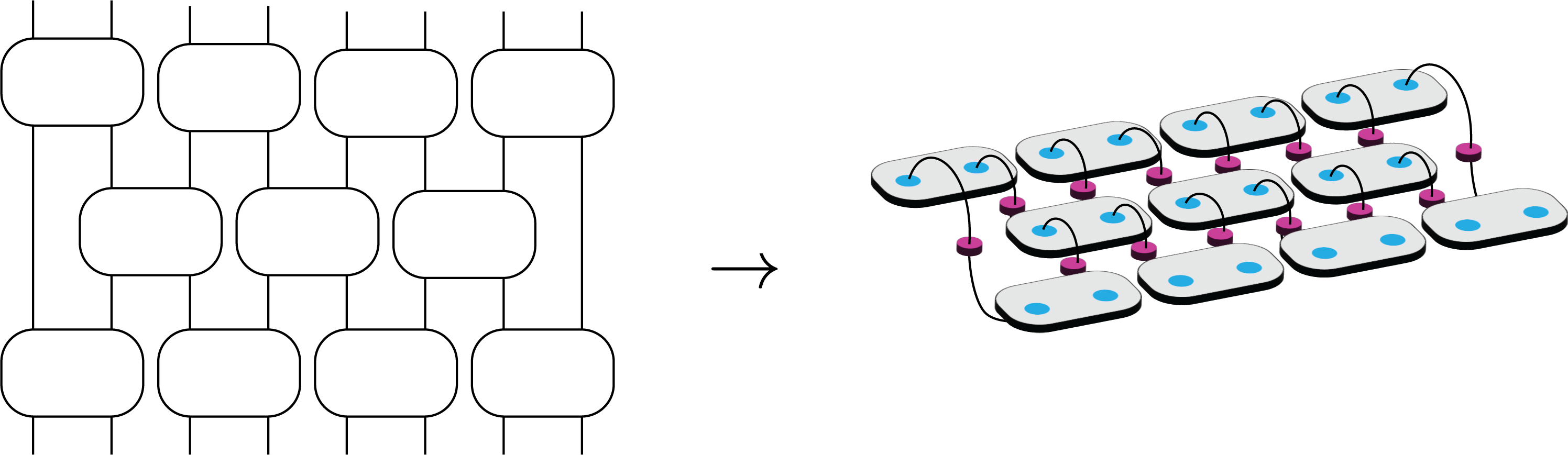}
    \caption{Graphical representation of Clifford-grid circuit parallelization.
    Every blue dot represents a qubit and all Clifford gates (boxes) are applied in parallel.
    The lines again represent Bell-state creations and Bell-state measurements, indicated by the pink boxes. 
    The propagating Pauli errors can be corrected using the Bell-state measurement results.}
    \label{fig:clifford_grid}
\end{figure}

Finding the matrix $M$ for correcting a Clifford-grid circuit is more complex than for a Clifford-ladder circuit. 
An error occurring in the grid can have multiple paths contributing to a single output wire. 
For the final correction, the parity of each contributing error-path is needed. 
This computation is in $\oplus\mathsf{L}\subseteq\NC^2$~\footnote{This is not too surprising as simulating Clifford circuits classically is a complete problem for $\oplus \mathsf{L}$}.
A precomputation again gives the matrix corresponding to the bilinear map $M$. 

\subsection{Useful gates and routines with an \texorpdfstring{$\LAQCC$}{LAQCC} implementation\label{sec:gates_created_in_LAQCC}}
This section groups useful multi-qubit gates with an $\LAQCC$ implementation.
The construction of $W$-states and Dicke states uses these gates, but their use might be of broader interest. 

The tables give the action of the gates on computational basis states. 
Their effect on arbitrary states follows by linearity. 
The tables also give the width of the implementation and a reference to the implementation.

The first two gates directly follow from the Clifford-parallelization results described in Section~\ref{sec:clifford_circuits}.  
\begin{table}[htb]
\centering
\begin{tabular}{l|l|l|l}
\textbf{Gate} & \textbf{Operation on basis states} & \textbf{Width} & \textbf{Implementation}\\
\hline
$\text{Fanout}_{n}$ & $ \ket{x}\ket{y_1}\hdots\ket{y_{n}} \mapsto\ket{x}\ket{y_1\oplus x}\hdots\ket{y_{n}\oplus x}$ & $\mathO(n)$ & Clifford-ladder circuit~\ref{lem:clifford_ladder} \\
Permutation$(\pi)_n$ &  $\ket{y_1}\hdots\ket{y_n} \mapsto \ket{y_{\pi(1)}}\hdots\ket{y_{\pi(n)}}$ & $\mathO(n^2)$ & Clifford-grid circuit~\ref{lem:clifford_grid}
\end{tabular}
\caption{Operations contained in $\LAQCC$ based on Clifford parallelization. 
Here $\pi\in S_n$ denotes a permutation of $n$ elements.}
\label{tab:Fanout_Perm}
\end{table}

Prior works extensively studied the fanout gate, for instance to construct a constant-depth $\text{OR}_n$ function with one-sided error~\cite{HoyerSpalek:2005} and with an exact implementation~\cite[Theorem~1]{TakahashiTani_CollapseOfHierarchyConstantDepthExactQuantumCircuits_2013}, both assuming the fanout gate to be a native gate.
The $\text{OR}_n$ gate also implies two other gates, as the following table shows.
\begin{table}[htb]
\centering
\begin{tabular}{l|l|l|l}
\textbf{Gate} & \textbf{Operation on basis states} & \textbf{Width} & \textbf{Implementation}\\
\hline
$\text{OR}_{n}$ & $ \ket{y_1}\hdots\ket{y_{n}} \ket{x} \mapsto\ket{y_1}\hdots\ket{y_n}\ket{\text{OR}_n(y)\oplus x}$ & $\mathO(n\log(n))$ & Appendix~\ref{gate:OR_implementation} \\
$\text{AND}_n$ &  $\ket{y_1}\hdots\ket{y_{n}} \ket{x} \mapsto\ket{y_1}\hdots\ket{y_n}\ket{\text{AND}_n(y)\oplus x}$ & $\mathO(n\log(n))$ & negate input and output of $\text{OR}_n$\\  
$\text{Equal}_i$ &  $\ket{j}\ket{b} \mapsto \begin{cases} \ket{j}\ket{1\oplus b} & \text{if } \ket{j} = \ket{i} \\
    \ket{j}\ket{b} & \text{else}
    \end{cases}$ & $\mathO(n\log(n))$ & negate part of input of $\text{AND}_n$ 
\end{tabular}
\caption{Operations contained in $\LAQCC$ based on Fanout and local $1$-qubit unitaries.}
\label{tab:Or_And_Equal}
\end{table}

With these unbounded-fan-in OR and AND gates, all $\AC^0$ circuits can be implemented. 
The next step is implementing $\LAQCC$-type modular addition circuits, which gives circuits to check for equality and greater-than. 
These three gates take $n$-qubit quantum states as input. 
We introduce the indicator variable $\mathbbm{1}_{A}$ for a Boolean expression $A$, which evaluates to $1$ if $A$ is true. 
Similarly, $\ket{\mathbbm{1}_{A}}=\ket{1}$ if and only if $A$ is true. 
\begin{table}[htb]
\centering
\begin{tabular}{l|l|l|l}
\textbf{Gate} & \textbf{Operation on $n$-qubit integers $\ket{x}$, $\ket{y}$} & \textbf{Width} & \textbf{Implementation}\\
\hline
$\text{Add}_{n}$ & $ \ket{x}\ket{y} \mapsto\ket{x}\ket{y + x \bmod 2^n}$ & $\mathO(n^2)$ & $\AC^0$ circuit \\
Equality &  $\ket{x}\ket{y}\ket{0}  \mapsto\ket{x}\ket{y}\ket{\mathbbm{1}_{x = y}}$ & $\mathO(n^2)$ & Appendix~\ref{gate:equality}\\  
Greatherthan &  $\ket{x}\ket{y}\ket{0} \mapsto \ket{x}\ket{y}\ket{\mathbbm{1}_{x > y}}$ & $\mathO(n^2)$ & Appendix~\ref{gate:greaterthan}
\end{tabular}
\caption{Operations contained in $\LAQCC$ based on $\AC^0$ circuits.}
\label{tab:Add_Equality_Greaterthan}
\end{table}

\citeauthor{HoyerSpalek:2005} showed that fanout-gates imply efficient constant-depth implementations of for instance the quantum Fourier transform~\cite{HoyerSpalek:2005}.
They use this constant-depth quantum Fourier transform to construct a constant-depth circuit for weighted counting. 
In particular, this circuit can be used to calculate the Hamming weight of an $n$-bit string, and to implement an ``Exact~$t$"-gate and a threshold gate. 
Appendix~\ref{gate:threshold_t} also explains how to modify the threshold gate to a weighted threshold gate. 
 \begin{table}[htbp]
\centering
\begin{tabular}{l|l|l|l}
\textbf{Gate} & \textbf{Operation on $n$-qubit basis state $\ket{x}$ } & \textbf{Width} & \textbf{Implementation}\\
\hline
$\text{QFT}$ & $\ket{x} \mapsto \frac{1}{\sqrt{2^{n-1}}}\sum_{j=0}^{2^{n-1}} e^{i 2 \pi \frac{x \cdot j}{2^n}}\ket{j}$ & $\mathO(n^3\log(n))$& \cite[Theorem~4.12]{HoyerSpalek:2005} \\
Hammingweight & $ \ket{x}_{n}\ket{0}_{\log(n)} \mapsto \ket{x}_n\ket{|x|}_{\log(n)} $& $\mathO(n\log(n))$ & \cite[Lemma~4]{TakahashiTani_CollapseOfHierarchyConstantDepthExactQuantumCircuits_2013} \\  
$\text{Exact}_t$ &  $\ket{x}\ket{0} \mapsto \ket{x}\ket{\mathbbm{1}_{|x| = t}}$ & $\mathO(n\log(n))$ & Appendix~\ref{gate:exact_t} \\
$\text{Threshold}_{t}$ &  $\ket{x}\ket{0} \mapsto \ket{x}\ket{\mathbbm{1}_{\sum_i x_i \ge t}}$ & $\mathO(t n \log(n))$ & Appendix~\ref{gate:threshold_t}
\end{tabular}
\caption{Quantum subroutines in $\LAQCC$ based on \citeauthor{HoyerSpalek:2005}.}
\label{tab:QFT_Hammingweight_Threshold}
\end{table} 

\begin{remark}
    As the $\text{Threshold}_{t}$ gate is in $\LAQCC$, any classical $\TC_0$ circuit is in $\LAQCC$.
\end{remark}

This section concludes not with a gate, but with a tool used for preparing uniform superpositions. 
This lemma extends Lemma~\ref{lem:exact_grover} to preparing states instead of finding marked items. 
\begin{lemma}\label{lem:grover_constant_fraction}
Given an $n$-qubit unitary $U$, that is implementable by a constant-depth circuit, a basis $\mathcal{C}$ and a partition of $\mathcal{C}$ in $\mathcal{G}$ and $\mathcal{B}$ such that $\frac{|\mathcal{G}|}{|\mathcal{C}|}$ is a known constant $c$. 
Suppose that $U$ implements the map
$$U:\ket{y}\ket{b}\mapsto\begin{cases}
\ket{y}\ket{b\oplus 1} & \qquad \text{if }y\in\mathcal{G} \\
\ket{y}\ket{b} & \qquad \text{if }y\in\mathcal{B} 
\end{cases}.$$
Then there exists a $\LAQCC$ circuit that prepares the state $\frac{1}{\sqrt{|\mathcal{G}|}}\sum_{y\in\mathcal{G}} \ket{y}$ by using $U$ a constant number of times.
\end{lemma}
\begin{proof}
Define $\ket{\mathcal{G}}=\frac{1}{\sqrt{|\mathcal{G}|}}\sum_{y\in\mathcal{G}}\ket{y}$ and $\ket{\mathcal{B}}=\frac{1}{\sqrt{|\mathcal{B}|}}\sum_{y\in\mathcal{B}}\ket{y}$. 
As $\mathcal{B}$ and $\mathcal{G}$ partition $\mathcal{C}$, it follows that $\braket{\mathcal{G}}{\mathcal{B}}=0$. Lemma~\ref{lem:exact_grover} implies the existence of a circuit that prepares the desired state. Below, we explicitly construct the circuit. 

First, prepare a uniform superposition $\sum_{i=0}^{2^n-1}\ket{i}$. 
Then, iteratively reflect over the state $\ket{\mathcal B}$ using $U$, and reflect over the uniform superposition state $\sum_{i=0}^{2^n-1}\ket{i}$. 
Both reflections have a $\LAQCC$ implementation and we only need to apply them a constant number of iterations. 

To reflect over the uniform superposition, we have to implement the operation $2\kb{s} - I$, with $\ket{s}=\frac{1}{\sqrt{N}}\sum_{i=0}^{2^n-1}\ket{i}$. 
To implement this operation, we first apply a layer of Hadamards, which implements a basis transformation mapping the uniform superposition state to the all zeroes state; Then apply the $\text{Exact}_0$-gate producing an output qubit that marks only the all zeroes-state and finally negate the output qubit and applies a $Z$-gate on it. Running this circuit in reverse, excluding the $Z$-gate, resets the output qubit and reverts the basis transformation. 
The last step of Lemma~\ref{lem:grover_constant_fraction} requires a reflection using an $R_Z$-gate (rotational $Z$-gate) instead of the $Z$-gate.
As the $\text{Exact}_0$-gate has an $\LAQCC$ implementation (see Table~\ref{tab:QFT_Hammingweight_Threshold}), this second inversion operation has a $\LAQCC$ implementation. 

The total number of iterations is $\mathO(\sqrt{N/m})$, where $N=|\mathcal{C}|$ and $m=|\mathcal{G}|$. 
As their fraction is the constant $c$, it follows that $\mathO(\sqrt{c})=\mathO(1)$ iterations are needed. 
\end{proof}

\subsection{Non-simulatability of \texorpdfstring{$\LAQCC$}{LAQCC}}
\label{sec:IQP_in_LAQCC}
Most of the power of $\LAQCC$ circuits seems to come from the classical intermediate calculations, which makes one wonder if these circuits are classically simulatable. 
Even if these circuits were indeed efficiently simulatable, they still have value as ``fast'' alternatives for state preparation.
However, it is unlikely that all $\LAQCC$ circuits can be simulated efficiently by a classical simulator. 
Lemma~\ref{lem:unitar_parallelization} and the inclusion of the fan-out gate in $\LAQCC$ show that circuits consisting of commuting gates have an $\LAQCC$ implementation and in particular, the class of Instantaneous Quantum Polynomial-time ($\mathsf{IQP}$) circuits, first introduced in~\cite{Shepherd2009}, has equivalent $\LAQCC$ implementations. 

\begin{definition}[Definition 2~\cite{Nakata_2014}]
An $\mathsf{IQP}$ circuit on $n$ qubits is a quantum circuit with the following structure: each gate in the circuit is diagonal in the Pauli-$Z$ basis, the input state is $\ket{+}^{\otimes n}$, and the output is the result of a measurement in the Pauli-X basis on a specified set of output qubits.
\end{definition}
\begin{lemma}
Any $\mathsf{IQP}$ circuit has an $\LAQCC$ implementation.
\end{lemma}
\begin{proof}
The following $\LAQCC$ circuit prepares the desired state:
First prepare $\ket{+}^{\otimes n}$ by a single layer of Hadamard gates on all qubits. 
In this basis, all gates in with respect to the Pauli-$Z$ basis commute, and hence by Lemma~\ref{lem:unitar_parallelization}, we can parallelize all gates using poly$(n)$ ancilla qubits. 
Next, we can again apply a layer of Hadamard gates and finally measure the desired qubits. 
\end{proof}

\citeauthor{Bremner2010} showed that efficient weak classical simulation of all possible IQP circuits up to small multiplicative error implies a collapse of the polynomial hierarchy~\cite{Bremner2010}. Note that a circuit family is weakly simulatable if given the description of the circuit family, its output distribution can be sampled by purely classical means in poly$(n)$ time.

\begin{lemma}[Corollary 1~\cite{Bremner2010}]
If the output probability distributions generated by uniform families of $\mathsf{IQP}$ circuits could be weakly classically simulated to within multiplicative error $1 \leq c < \sqrt{2}$ then the polynomial hierarchy would collapse to the third level, in particular, $\mathsf{PH} = \Delta^p_3$.
\end{lemma}

\begin{corollary}
If the output probability distributions generated by uniform families of $\LAQCC$ circuits could be weakly classically simulated to within multiplicative error $1 \leq c < \sqrt{2}$ then the polynomial hierarchy would collapse to the third level, in particular, $\mathsf{PH} = \Delta^p_3$.
\end{corollary}

\subsection{\texorpdfstring{$\LAQCC$}{LAQCC} containment in \texorpdfstring{$\QNC^1$}{QNC1}}
\label{sec:LAQCC_in_QNC1}
Let $A$ be an $\LAQCC$-circuit. 
We can write this circuit as a composition of unitary quantum layers $U_i$, measurements $M_i$ and classical calculation layers $C_i$:
$$A = M_k U_k C_k \dots M_i U_i C_i \dots M_1 U_1 C_1,$$
for some constant $k$. 
Any unitary $U_i$ is a $\QNC^0$ circuit and any $C_i$ is an $\NC^1$-circuit. 
The measurements $M_i$ can measure any subset of the qubits.
By the principle of deferred measurements, we can always postpone them to the end of the circuit using $CNOT$ gates and fresh ancilla qubits~\cite[Section~4.4]{nielsen_chuang_2010}, which gives the following lemma. 
\begin{lemma}
\label{lem:LAQCC_QNC1}
For any $\LAQCC$-circuit $A$ there is a $\QNC^1$-circuit $B$ without intermediate measurements that outputs the same state as $A$.
\end{lemma}
\begin{proof}
The $\LAQCC$-circuit $A$ contains classical computation layers $C_i$ that use the intermediate measurement results as input. 
These measurements can be delayed until the end of the circuit by applying a CNOT from the qubit to a fresh ancilla qubit.
This replaces the classical output wires by quantum wires. 

Lemma~\ref{lem:NC1toQauntum} shows that any $\NC^1$-circuit can be run on a log-depth quantum circuit with $\mathO(\mathrm{poly}(n))$ ancilla qubits. 
Hence, a $\QNC^1$-circuit without topology constraints can take the role of the classical intermediate circuits $C_i$.

Now, let $V_i$ be the quantum circuit implementing $C_i$.
Then the quantum circuit
$$B = U_k V_k \dots U_1 V_1$$
is a quantum circuit of logarithmic depth simulating $A$.
\end{proof}

\subsection{Complexity results for \texorpdfstring{$\LAQCC(\mathcal{Q}, \mathcal{C},d)$}{LAQCC(Q,C,d)}}
\label{sec:Complexity results}

In the definition of $\LAQCC(\mathcal{Q}, \mathcal{C},d)$, we have freedom to choose $\mathcal{Q}$ and $\mathcal{C}$. 
If we give more power to both the quantum and the classical routines, we see that we can solve more complex problems and prepare a wider variety of quantum states. 
Yet, even with polynomial depth quantum circuits and unbounded classical computational power, limits exist.

\begin{notation}
The class $\LAQCC^*$ is the instantiation $\LAQCC(\mathsf{QPoly}(n), \mathsf{ALL}, \mathrm{poly}(n))$:
The class of polynomially many alternating polynomial-sized quantum circuits and arbitrary powerful classical computations, together with feed-forward of the classical information to future quantum operations. 
The quantum computations are restricted to all single-qubit gates and the two-qubit CNOT gate. 
\end{notation}
This directly gives us a definition of $\mathsf{State}$$\LAQCC^*_{\varepsilon}$ for $\varepsilon>0$. 
\begin{remark}
Note that for any non-zero $\varepsilon$, we can restrict ourselves to finite universal gate-sets. 
The Solovay-Kitaev theorem~\cite{Kitaev1997,nielsen_chuang_2010} says that any multi-qubit unitary can be approximated to within precision $\delta$ by a quantum circuit with size depending on $\delta$. 
Therefore, with a finite universal gate-set, any $\LAQCC^*$ circuit with a continuous gate-set can be approximated by an $\LAQCC^*$ circuit with gates from the finite set. 
\end{remark}

Next, we prove $\mathsf{State}$$\LAQCC^*_{\varepsilon}\subseteq\mathsf{StatePostQPoly}_{\varepsilon}$. 
Following the same decomposition as in the previous section, we find that any $\LAQCC^*$ can be written as 
\begin{equation}
    \Pi_{i=0}^{\text{poly}(n)} M_i U_i(y_i) C_i(x_i)\ket{0}^{\otimes \text{poly}(n)},
\end{equation}
where again, $M_i$ denotes the $i$-th measurement layer, $U_i$ the $i$-th quantum layer and $C_i$ the $i$-th unbounded classical computation layer. 
The $x_i\in\{0,1\}^*$ denote the outcomes of $M_i$ and $y_i\in \{0,1\}^*$ the bitstring outputted by $C_i$. 
Note, all $x_i$ and $y_i$ have length at most polynomial in $n$. 

\begin{theorem}
\label{thm:state_prep_post-BQP}
It holds that $\mathsf{State}$$\LAQCC^*_{\varepsilon}\subseteq\mathsf{StatePostQPoly}_{\varepsilon}$.
\end{theorem}
\begin{proof}
Fix $\varepsilon>0$ and a positive integer $n$ and 
let $\ket{\psi}\in\mathsf{State}$$\LAQCC^*_{\varepsilon}$. 
By definition, there exists an $\LAQCC^*$ circuit $A=\Pi_{i=0}^{\text{poly}(n)} M_i U_i(y_i) C_i(x_i)$, which prepares a state $\ket{\phi}$ with inner product at least $1-\varepsilon$ with $\ket{\psi}$.

Then consider the following $\mathsf{PostQPoly}$-circuit:
Let $B=\Pi_{i=0}^{\text{poly}(n)} \text{Equal}_{x_i}(x_i)U_i(y_i)\ket{0}^{\otimes \text{poly}(n)}$, where the $y_i$ are hardwired. 
The Equal$_{x_i}$ gate replaces the measurement layer $M_i$, by checking if the subset of qubits that would be measured are in $\ket{x_i}$ computational basis state. It stores the output in an ancilla qubit. 
As a last step, apply an $\text{AND}_{\text{poly}(n)}$-gate on the ancilla qubits, which hold the outputs of the Equal$_{x_i}$ gates, and store the result in an ancilla qubit. 
Conditional on this last ancilla qubit being one, the circuit prepares the state $\ket{\phi}$.
\end{proof}

Figure~\ref{fig:laqcc_to_post_bqp} gives a schematic overview of the proof and the translation of an $\LAQCC^*$ circuit in a $\mathsf{PostQPoly}$-circuit. 
\begin{figure}[h]
\centering
\includegraphics[width=0.8\textwidth]{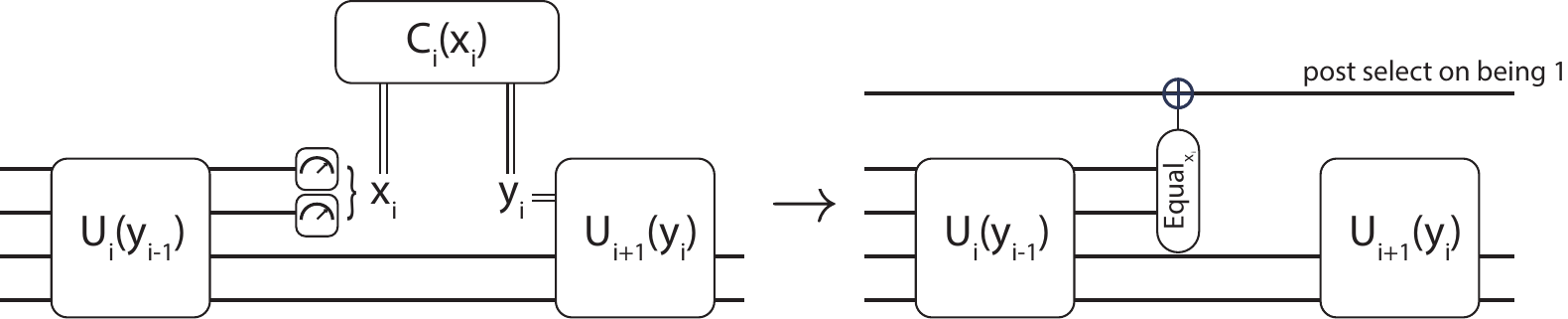}
\caption{Schematic idea of transforming an $\LAQCC^*$ circuit for generating $\ket{\psi}$ into a $\mathsf{PostQPoly}$ circuit.}
\label{fig:laqcc_to_post_bqp}
\end{figure}

\paragraph{Further complexity studies of \texorpdfstring{$\LAQCC$}{LAQCC}:}
An interesting direction for future work is to extend the inclusion proved above to a true separation or an oracular separation. 
One approach is to use a similar oracle as used in~\cite{AaronsonKuperberg:2007} to separate $\mathsf{QMA}$ and $\mathsf{QCMA}$ with respect to an oracle and use a counting argument to argue that $\mathsf{State}$${\LAQCC^*_{\varepsilon}}^{\mathcal{O}}\neq \mathsf{StatePostQPoly}_{\varepsilon}^{\mathcal{O}}$, for some oracle $\mathcal{O}$ and $\varepsilon=1-\frac{1}{\text{poly}(n)}$.

Another interesting direction for studying the state complexity of $\LAQCC$ is in the fact that the $\LAQCC$ model allows for a constant number, more than one, of rounds of measurements and corrections. This was required for our three new state generation protocols. However other models considered only one round of measurements and corrections, for instance in the paper \cite{Cirac:2021}. One may wonder if there is a hierarchy of model power allowing one or multiple measurements, and if there is a way to reduce the number of measurements rounds. A starting effort towards classifying types of states based on such a hierarchy can be found in \cite{Tantivasadakarn_2023}. It would be interesting to see a more extensive complexity theoretic analysis comparing models with different number of allowed rounds.

%% file: Uniform_superposition.tex
\subsection{Uniform superposition of size \texorpdfstring{$q$}{q}\label{sec:superposition_modulo_q}}
The uniform superposition is often used as initial state in other quantum algorithms. 
A simple Hadamard gate applied to $n$ qubits prepares the uniform superposition $\frac{1}{\sqrt{2^n}}\sum_{i=0}^{2^n-1}\ket{i}$.
Preparing the state $\frac{1}{\sqrt{q}}\sum_{i=0}^{q-1}\ket{i}$, the superposition up to size $q$, is already harder for arbitrary $q$. 

A simple probabilistic approach works as follows: 
1) create a superposition $\frac{1}{\sqrt{2^n}}\sum_{i=0}^{2^n-1}\ket{i}$  with $n = \ceil{\log_2(q)}$ qubits;
2) mark the states $i< q$ using an ancilla qubit;
3) measure this ancilla qubit.
Based on the measurement result, the desired superposition is found, which happens with probability at least one half. 

The next theorem modifies this probabilistic approach to a protocol that deterministically prepares the uniform superposition modulo $q$ in $\LAQCC$.
\begin{theorem}
\label{thm:uniform_superposition_mod_q}
There is a deterministic $\LAQCC$ circuit that prepares the uniform superposition of size $q$. This circuit requires $\mathO(\ceil{\log_2(q)}^2)$ qubits.
\end{theorem}
\begin{proof}
Let $n = \ceil{\log_2(q)}$ and define $\mathcal{G}=\{i\mid 0\le i<q\}$ and $\mathcal{B}=\{i\mid q\le i\le2^n-1\}$. 
Construct the unitary
$$U_q:\ket{y}\ket{b}\mapsto\begin{cases}
\ket{y}\ket{b\oplus 1} & \text{if } y<q \\
\ket{y}\ket{b} & \text{if } y\ge q
\end{cases}.$$
The Greaterthan-gate of Table~\ref{tab:Add_Equality_Greaterthan} implements the operator $U_q$, note that this gate requires $\mathO(n^2)$ qubits. 

As $|\mathcal{G}|/2^n \ge 1/2$ and known, applying Lemma~\ref{lem:grover_constant_fraction} with the sets $\mathcal{G}$ and $\mathcal{B}$ and the constant-depth implementation of $U_q$, gives a $\LAQCC$ algorithm that boosts the amplitude of $\ket{\mathcal{G}}$ to~$1$.
\end{proof}

\begin{remark}
Note that in Lemma~\ref{lem:grover_constant_fraction} it was implicitly assumed that $|\mathcal G| + |\mathcal B|$ is a power of two (allowing for a simple reflection over the uniform superposition state). This $\LAQCC$ implementation of creating a uniform superposition modulo any $q$ removes this requirement.
\end{remark}

%% file: W_states.tex
\subsection{\texorpdfstring{$W$}{W}-state in \texorpdfstring{$\LAQCC$}{LAQCC}}
\label{sec:W_state_in_LAQCC}
In this section we consider the $W_n$-state and how to prepare this state in $\LAQCC$. 
The $W_n$-state is a uniform superposition over all $n$-qubit states with a single qubit in the $\ket{1}$-state and all others in the $\ket{0}$-state:
$$\ket{W_n}  = \frac{1}{\sqrt{n}}\sum_i \ket{e_i},$$
where $\ket{e_i}$ is the state with a one on the $i$-th position and zeroes elsewhere. 

A first observation is that the $W$-state can be seen as a one-hot encoding of a uniform superposition over $n$ elements. 
We can label the $n$ states with non-zero amplitude of the $W$-state with an index. More precisly, we want to find circuits that implement the following map:
\begin{align}
\label{eqn:i_to_ei}
    \ket{i}\ket{0}\mapsto \ket{0}\ket{e_i},
\end{align}
with $i$ an index and $e_i$ the one-hot encoding of $i$.
This index -- which equals the position of the $1$ -- compresses the representation from $n$ to $\log(n)$ bits. 
This compression naturally defines two operations: 
\begin{align}
\text{\textbf{Uncompress}: }& \ket{i}_{\log(n)}\ket{0}_{n} \mapsto \ket{i}_{\log(n)}\ket{e_i}_{n}, \\
\text{\textbf{Compress}: }& \ket{i}_{\log(n)}\ket{e_i}_{n} \mapsto \ket{0}_{\log(n)}\ket{e_i}_{n}.
\end{align}
Implementing both and combining them implements Mapping~\ref{eqn:i_to_ei} giving an efficient $W$-state preparation protocol. 

The \textbf{Compress} and \textbf{Uncompress} operations map between a one-hot and binary representation of an integer $i$. We call the registers containing the binary representation index registers, and the register containing the one-hot representation the system register. The index registers serve as ancilla qubits and the $W$-state is prepared in the system register. 

\begin{lemma}
\label{lem:uncompress}
There exists an $\LAQCC$ circuit that, for any $n$, implements the \textbf{Uncompress} operation:
\[
\frac{1}{\sqrt{n}}\sum_{i = 0}^{n-1}\ket{i}_{\log(n)}\ket{0}_n \mapsto \frac{1}{\sqrt{n}}\sum_{i = 0}^{n-1}\ket{i}\ket{e_i}_n.
\]
This circuit uses $\mathO(n \log(n))$ qubits placed in a grid pattern of size $n\times \log(n)$.
\end{lemma}
\begin{proof}
One column of the grid of length $n$ consists of system qubits placed in a line.
Adjacent to this line are $\log(n)$ index qubits. 
The left grid in Figure~\ref{fig:W_state_uncompress} shows the initial layout. 
The same figure also shows the steps to prepare the $W$-state in the system qubits. 
\begin{align*}
\frac{1}{\sqrt{n}}\sum_{i = 0}^{n-1}\ket{i}_{\log(n)}\ket{0}_{\log(n)}^{\otimes n-1}\ket{0}_n & \xrightarrow{(1)} \frac{1}{\sqrt{n}}\sum_{i = 0}^{n-1}\ket{i}_{\log(n)}^{\otimes n}\ket{0}_n \\
    								   & \xrightarrow{(2)} \frac{1}{\sqrt{n}}\sum_{i = 0}^{n-1}\ket{i}_{\log(n)}^{\otimes n}\ket{e_i}_n \\
     								   & \xrightarrow{(3)} \frac{1}{\sqrt{n}}\sum_{i = 0}^{n-1}\ket{i}\ket{0}^{\otimes n-1}\ket{e_i}_n 
\end{align*}
Step (1) uses fanout-gates to create a fully entangled state between the different index registers.
Step (2) applies $\text{Equal}_i$ gates in parallel from each index register to its corresponding system qubit to create the state $\ket{e_i}$ in the system register. 
Step (3) uses fanout-gates to disentangle and reset the index registers. 
Combined the \textbf{Uncompress} operation maps  $\frac{1}{\sqrt{n}}\sum_{i = 0}^{n-1}\ket{i}_{\log(n)}\ket{0}_{\log(n)}^{\otimes n-1}\ket{0}_n \mapsto \frac{1}{\sqrt{n}}\sum_{i = 0}^{n-1}\ket{i}\ket{0}^{\otimes n-1}\ket{e_i}_n$ as required.
\end{proof}
\begin{figure}[h!]
\includegraphics[width=\textwidth]{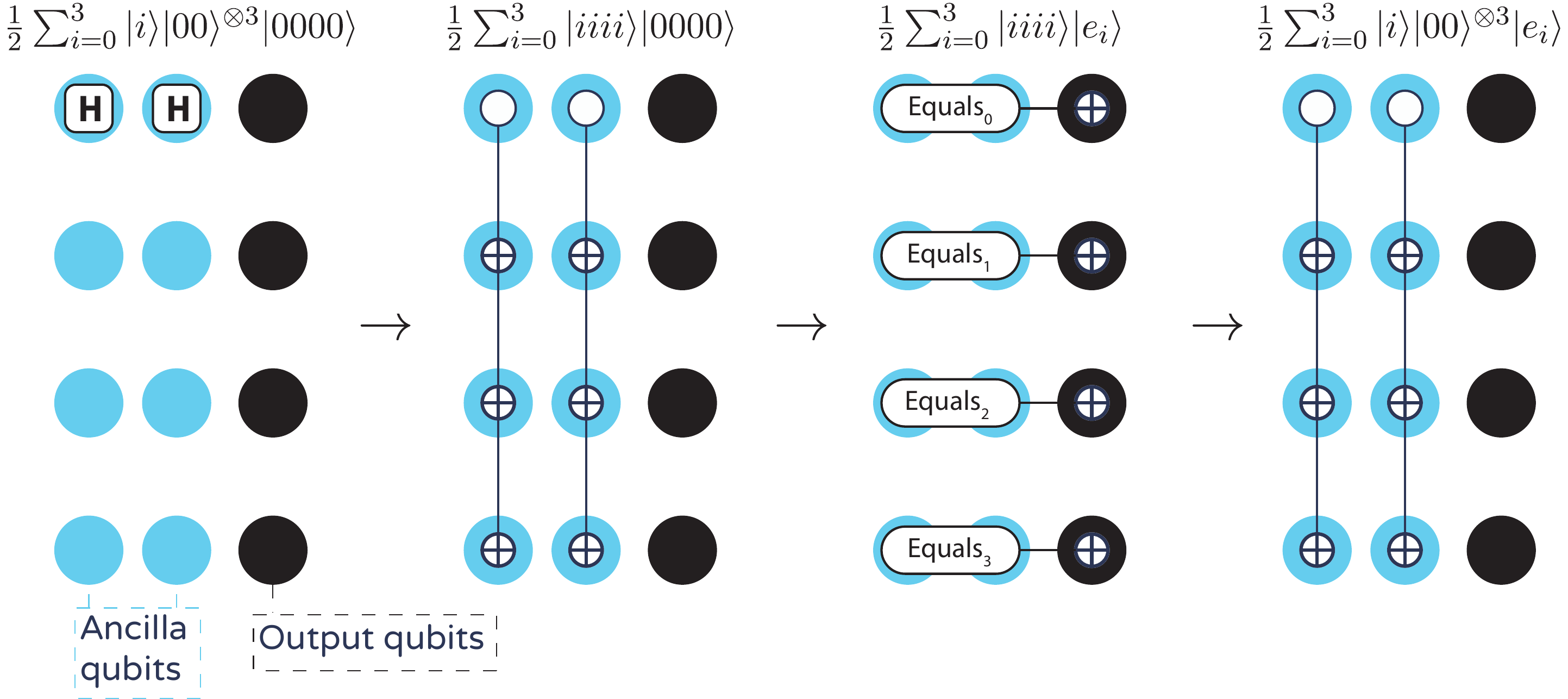}
\caption{Circuit for the \textbf{Uncompress} operation for $n=4$. 
Shown is a grid of $12$ qubits: $8$ blue index qubits, and $4$ black system qubits. 
Each of the four grids represents a single timeslice in the \textbf{Uncompress} operation.}
\label{fig:W_state_uncompress}
\end{figure}

\begin{lemma}
\label{lem:compress}
There exists an $\LAQCC$ circuit that, for any $n$, implements the \textbf{Compress} operation:
\[
\frac{1}{\sqrt{n}}\sum_{i = 0}^{n-1}\ket{i}_{\log(n)}\ket{e_i}_n \mapsto \frac{1}{\sqrt{n}}\sum_{i = 0}^{n-1}\ket{0}\ket{e_i}_n.
\]
This circuit uses $\mathO(n \log(n))$ qubits placed in a grid pattern of size $n\times \log(n)$.
\end{lemma}
\begin{proof}
To implement \textbf{Compress}, the index registers are uncomputed using parallel $CNOT$-operations, controlled by the system register. 
These controlled gates commute for different indices in the system register and hence by Lemma~\ref{lem:unitar_parallelization} a parallel circuit for the uncomputation exists. The \textbf{Compress} operation, also shown in Figure~\ref{fig:W_state_compress}, consists of the operations:
\begin{align*}
\frac{1}{\sqrt{n}}\sum_{i = 0}^{n}\ket{i}_{\log(n)}\ket{0}_{\log(n)}^{\otimes n - 1}\ket{e_i}_n & \xrightarrow{(1)} \frac{1}{n}\sum_{i, j = 0}^{n}(-1)^{i \cdot j}\ket{j}_{\log(n)}\ket{0}_{\log(n)}^{\otimes n - 1}\ket{e_i}_n \\
    & \xrightarrow{(2)} \frac{1}{n}\sum_{i, j = 0}^{n}(-1)^{i \cdot j}\ket{j}_{\log(n)}^{\otimes n}\ket{e_i}_n \\
    & \xrightarrow{(3)} \frac{1}{n}\sum_{i, j = 0}^{n}\ket{j}_{\log(n)}^{\otimes n}\ket{e_i}_n \\
    & \xrightarrow{(4)} \frac{1}{n}\sum_{i, j = 0}^{n}\ket{j}_{\log(n)}\ket{0}_{\log(n)}^{\otimes n-1}\ket{e_i}_n \\
    & \xrightarrow{(5)} \frac{1}{\sqrt{n}}\sum_{i =0 }^{n}\ket{0}_{\log(n)}^{\otimes n}\ket{e_i}_n 
\end{align*}
Step (1) applies Hadamard gates to the first index register, changing from the computational to the Hadamard basis, in which the $NOT$-operation is diagonal;
Step (2) uses fanout-gates to create a fully entangled state in the index registers;
Step (3) applies controlled-$Z$ gates, controlled by the system qubit $i$ and with targets the qubits in the $i$-th index register corresponding to the ones in the binary representation of $i$;
Step (4) disentangles the index registers using fanout-gates; 
and, Step (5) applies Hadamard gates to clean the index register.

The controlled-$Z$ gates in Step (3) apply phases that precisely cancel the phases already present, which disentangles the index registers from the system register. 
\end{proof}
\begin{figure}[h]
\includegraphics[width=\textwidth]{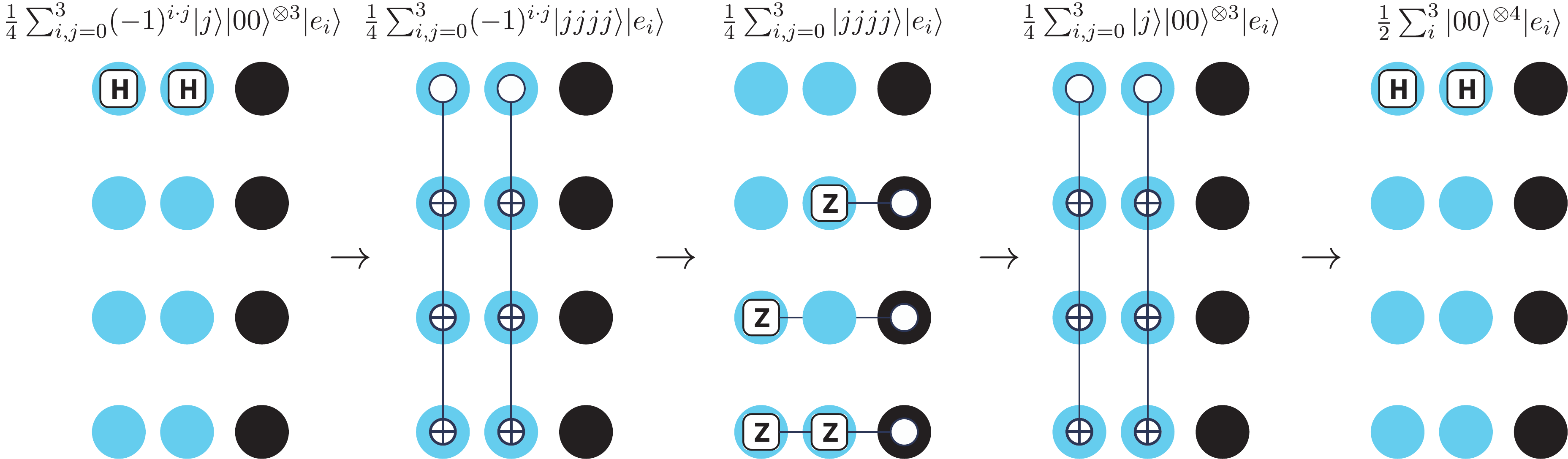}
\caption{Circuit for the \textbf{Uncompress} operation for $n=4$. 
Shown is a grid of $12$ qubits: $8$ blue index qubits, and $4$ black system qubits. 
Each of the four grids represents a single timeslice in the \textbf{Compress} circuit.}
\label{fig:W_state_compress}
\end{figure}

\begin{theorem}
\label{thm:W_state}
There exists a circuit in $\LAQCC$ that prepares the $\ket{W_n}$ state. This circuit requires $\mathO(n\log(n))$ qubits placed in a grid of size $n\times \log(n)$.
\end{theorem}

\begin{proof}
The circuit combines the circuits of Theorem~\ref{thm:uniform_superposition_mod_q}, Lemma~\ref{lem:uncompress} and Lemma~\ref{lem:compress}.
It consists of three steps:
\begin{align*}
\ket{0}^{\otimes n}_{\log(n)}\ket{0}_n &\xrightarrow[]{(1)} \frac{1}{\sqrt{n}}\sum_{i = 0}^{n-1}\ket{i}\ket{0}^{\otimes n-1}\ket{0}\\
        &\xrightarrow[]{(2)} \frac{1}{\sqrt{n}}\sum_{i = 0}^{n-1}\ket{i}\ket{0}^{\otimes n-1}\ket{e_i} \\
        &\xrightarrow[]{(3)}\frac{1}{\sqrt{n}}\sum_{i =0 }^{n}\ket{0}^{\otimes n}\ket{e_i} 
\end{align*}
Step one prepares the uniform superposition over indices, this can be done either by applying a layer of Hadamard gates, if $n$ is a power of $2$, requiring $\mathO(log(n))$ qubits,  or using Theorem~\ref{thm:uniform_superposition_mod_q} if $n$ is not a power of $2$ requiring $\mathO(\ceil{\log(n)}^2)$ qubits; Step (2) is by Lemma~\ref{lem:uncompress} and requires $\mathO(n\log(n))$ qubits; and, Step(3) is by Lemma~\ref{lem:compress} and requires $\mathO(n\log(n))$ qubits. 
\end{proof}

%% file: Dicke_states.tex
\subsection{Dicke states for small \texorpdfstring{$k$}{k}}
\label{sec:dicke:small_k}
In this section we generalize our method of preparing the $\ket{W}$-state in Theorem~\ref{thm:W_state} to a more general set of states, Dicke states. 
A Dicke-$(n,k)$ state is the uniform superposition over bitstrings of Hamming weight $k$ and length $n$ (which we again assume to be a power of $2$ for simplicity): 
\begin{align}
    \ket{D_k^n} = \binom{n}{k}^{-1/2}\sum_{x \in \{0,1\}^n: |x| = k} \ket{x}.
\end{align}
For $k=1$, this state is precisely the $W$-state. 
There exists an efficient deterministic method to prepare a $\ket{D_k^n}$ state that requires a circuit of width $\mathO(n)$ and depth $\mathO(n)$, independent of $k$~\cite{bartschi2019deterministic}. 
This methods starts from the $\ket{1}^{\otimes k}\ket{0}^{\otimes k - n}$ state and relies on a recursive formula for the Dicke state
\begin{align*}
    \ket{D_k^n} = \alpha_{k,n} \ket{D_k^{n-1}}\otimes \ket{0} + \beta_{k,n} \ket{D_{k-1}^{n-1}}\otimes \ket{1}.
\end{align*}
This relation implies a protocol that is inherently sequential, which is unsuited for an $\LAQCC$ implementation. 

Instead, we present an $\LAQCC$ approach similar to the $W$-state preparation protocol. 
We apply the \textbf{Uncompress} operation of the $W$-state in parallel to put $k$ ones into the bitstring. 
This method allows for the preparation of Dicke states with $k=\mathO(\sqrt{n})$, using $\mathO(n^2 \log(n)^3)$ qubits. The bound on $k$ comes from the fact that using the \textbf{Uncompress} operation in parallel might cause overlaps to where the $1$'s are put into the system register. Having two ones in the same system qubit in effect negates the \textbf{Uncompress} operation. 
Following the lines of the birthday paradox, we find that overlaps between different indices happen not that often for $k = \mathO(\sqrt{n})$.
Lemma~\ref{lem:grover_constant_fraction} allows us to boost the amplitudes and make the protocol deterministic.

Again, consider two groups of qubits: Index registers with $\log(n)$ qubits each; 
and, system registers of $n$ qubits each. 
Contrary to the $W$-state, the Dicke state requires multiple system registers during the preparation. 
The state is prepared in only one system register. 
Denote the index registers with subscripts $i_1$ up to $i_k$ and the system registers with $s_1$ up to $s_n$. 
For clarity, these indices may be omitted if it is clear from the context. 

The algorithm consists of four steps:
\begin{enumerate}
    \item \textbf{Filling}: $\ket{0}_{i_1}\dots\ket{0}_{i_k}\ket{0}_{s_1} \rightarrow \frac{1}{n^{k/2}}\sum_{j_1, \dots, j_k = 0}^{n-1} \ket{j_1}_{i_1} \dots \ket{j_k}_{i_k}\ket{e_{j_1} \oplus \dots \oplus e_{j_k} }_{s_1}$
    \item \textbf{Filtering}: $\rightarrow \sqrt{\frac{(n-k)!}{n!}}\sum_{j_1 \neq \dots \neq j_k}^{n-1} \ket{j_1} \dots \ket{j_k}\ket{e_{j_1} \oplus \dots \oplus e_{j_k}}$
    \item \textbf{Ordering}: $\rightarrow \frac{1}{\sqrt{\binom{n}{k}}}\sum_{j_1 < \dots < j_k}^{n-1} \ket{j_1} \dots \ket{j_k}\ket{e_{j_1} \oplus \dots \oplus e_{j_k}}$
    \item \textbf{Cleaning}: $\rightarrow \frac{1}{\sqrt{\binom{n}{k}}}\sum_{j_1 < \dots < j_k}^{n-1} \ket{0} \dots \ket{0}\ket{e_{j_1} \oplus \dots \oplus e_{j_k}}$
\end{enumerate}
Note that after \textbf{Filling} there is a multiplicity in states. First, \textbf{Filtering} removes those states in which different indices $j_l$ are the same, resulting in an incorrect state in the $s_1$ register. 
Second, \textbf{Ordering} removes the multiplicity from having multiple permutations of the index registers creating the same state in the $s_1$ register, by forcing a choice of ordering on the indices.
These two steps give a unique pairing between index registers and system registers allowing the operation \textbf{Cleaning}.

We will now proof that these four steps are achievable in $\LAQCC$ and explicitly visualize the corresponding circuits for $n=4$ and $k=2$. 

\begin{lemma}
\label{lem:Dicke_filling}
An $\LAQCC$ circuits exists that implements \textbf{Filling}:
$$\ket{0}_{i_1}\dots\ket{0}_{i_k}\ket{0}_{s_1} \rightarrow \frac{1}{n^{k/2}}\sum_{j_1, \dots, j_k = 0}^{n-1} \ket{j_1}_{i_1} \dots \ket{j_k}_{i_k}\ket{e_{j_1} \oplus \dots \oplus e_{j_k} }_{s_1}.$$ 
This circuit uses $\mathO(k n\log(n))$ qubits.
\end{lemma}
\begin{proof}
To achieve a circuit implementing \textbf{Filling} we use \textbf{Uncompress} from Lemma~\ref{lem:uncompress} $k$ times in parallel.
Note that two \textbf{Uncompress} operations commute, hence by Lemma~\ref{lem:unitar_parallelization} $k$ \textbf{Uncompress} operations can be implemented in parallel.
Each of these parallel operations requires an index register, a system register and $\mathO(n\log(n))$ extra ancilla qubits.

The corresponding circuit consists of five steps: 
\begin{align*}
\ket{0}_{i_1} \dots \ket{0}_{i_k} \ket{0}_{s_1}\dots\ket{0}_{s_k} &\xrightarrow{(1)} \frac{1}{n^{k/2}}\sum_{j_1\dots j_k = 0}^{n-1}\ket{j_1}\dots \ket{j_k} \frac{1}{\sqrt{2^n}}\sum_{l=0}^{2^n - 1}\ket{l}_{s_1} \ket{0}_{s_2} \dots \ket{0}_{s_k}
\\
&\xrightarrow{(2)} \frac{1}{n^{k/2}}\sum_{j_1\dots j_k = 0}^{n-1}\ket{j_1}\dots \ket{j_k} \frac{1}{\sqrt{2^n}}\sum_{l=0}^{2^n - 1}\ket{l}_{s_1} \ket{l}_{s_2} \dots \ket{l}_{s_k}
\\
&\xrightarrow{(3)}\frac{1}{n^{k/2}}\sum_{j_1\dots j_k = 0}^{n-1}\ket{j_1}\dots \ket{j_k} \frac{1}{\sqrt{2^n}}\sum_{l=0}^{2^n - 1}(-1)^{(2^{j_1} + \dots + 2^{j_k}) \cdot l}\ket{l}_{s_1} \ket{l}_{s_2} \dots \ket{l}_{s_k}
\\
&\xrightarrow{(4)} \frac{1}{n^{k/2}}\sum_{j_1\dots j_k = 0}^{n-1}\ket{j_1}\dots \ket{j_k} \frac{1}{\sqrt{2^n}}\sum_{l=0}^{2^n - 1}(-1)^{(2^{j_1} + \dots + 2^{j_k}) \cdot l}\ket{l}_{s_1} \ket{0}_{s_2} \dots \ket{0}_{s_k}
\\ 
&\xrightarrow{(5)} \frac{1}{n^{k/2}}\sum_{j_1\dots j_k = 0}^{n-1}\ket{j_1}\dots \ket{j_k} \ket{e_{j_1} \oplus \dots \oplus e_{j_k}}_{s_1} \ket{0}_{s_2} \dots \ket{0}_{s_k}
\end{align*}
Step (1) brings all index registers in a uniform superposition of size $n$, use Theorem~\ref{thm:uniform_superposition_mod_q} if required, and one system register in a uniform superposition size $2^n$; 
Step (2) uses fan-out gates to create entangled copies of the system register; 
Step (3) applies a phase flip between every pair of index and system register using \textbf{Uncompress} of Lemma~\ref{lem:uncompress}, except instead of applying not gates to the system registers, apply phase gates;
Step (4) uses fan-out gates to disentangle and uncompute all but one of the system registers; 
Step (5) applies Hadamard gates on the system register to obtain the one-hot representation of the index registers. 
Step(3), the step that requires most qubits, requires $\mathO(n \log(n))$ qubits for every pair of index and system register, of which there are $k$, resulting in the requirement of $\mathO(kn \log(n))$ qubits. 
\end{proof}

Figure~\ref{fig:Dicke_2_state} shows these five steps graphically. 
Ancilla qubits are omitted for clarity. 
Note that some of the $j_i$ in the index registers may intersect. 
The next Filtering step takes care of that.
\begin{figure}[ht]
\includegraphics[width=\textwidth]{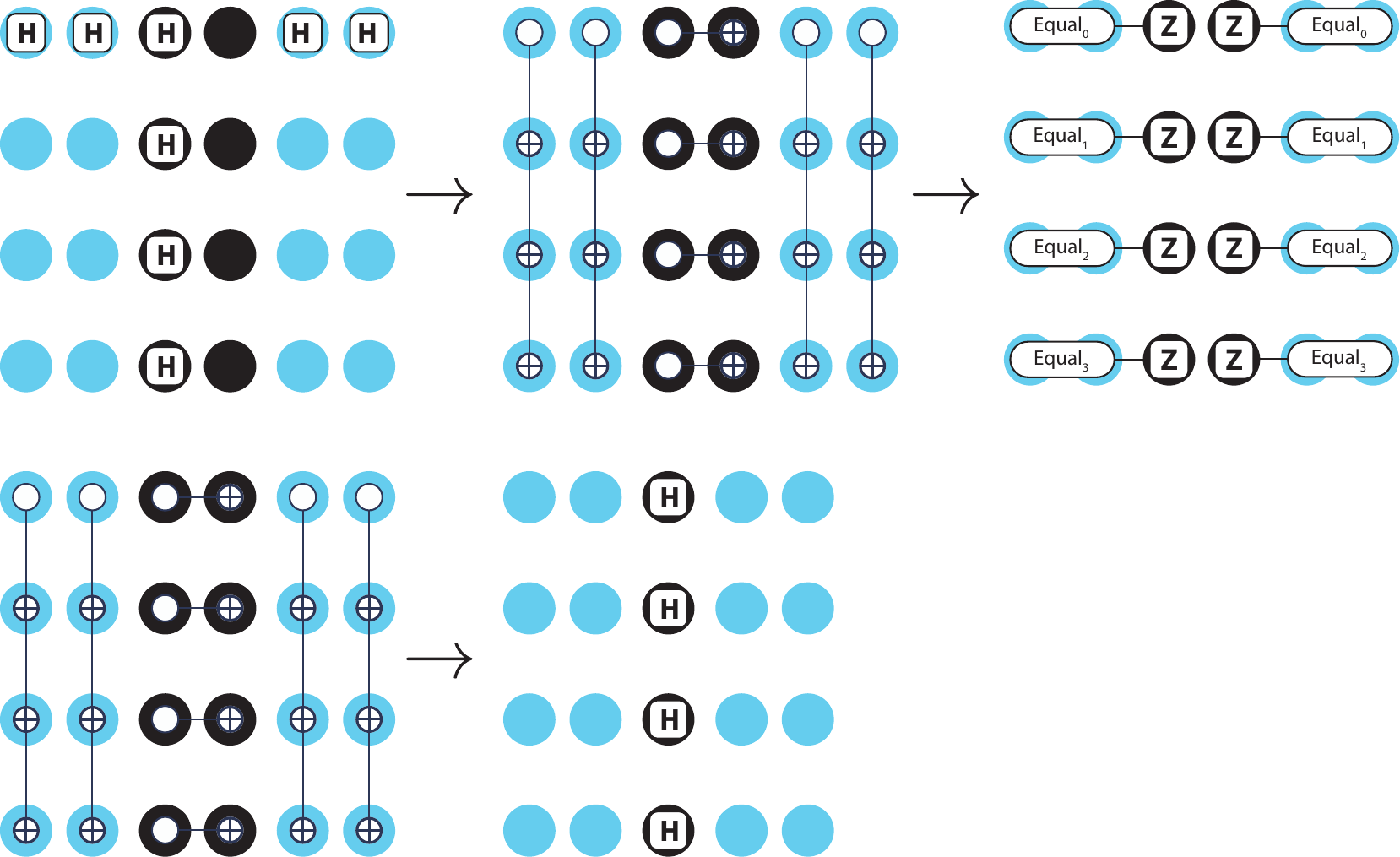}
\caption{Circuit to implement \textbf{Filling}, $\ket{0}_{i_1}\dots\ket{0}_{i_k}\ket{0}_{s_1} \rightarrow \frac{1}{n^{k/2}}\sum_{j_1, \dots, j_k = 0}^{n-1} \ket{j_1}_{i_1} \dots \ket{j_k}_{i_k}\ket{e_{j_1} \oplus \dots \oplus e_{j_k} }_{s_1}$. This circuit requires $\mathO(k n\log(n))$, for $n=4$ and $k = 2$. A grid of $24$ qubits is shown: $16$ blue index qubits and $8$ black system qubits. Each of the five grids represents a single timeslice in the circuit.}
\label{fig:Dicke_2_state}
\end{figure}

\begin{lemma}
\label{lem:Dicke_filtering}
An $\LAQCC$ circuit exists that implements \textbf{Filtering}:
$$\frac{1}{n^{k/2}}\sum_{j_1, \dots, j_k = 0}^{n-1} \ket{j_1}_{i_1} \dots \ket{j_k}_{i_k}\ket{e_{j_1} \oplus \dots \oplus e_{j_k} }_{s_1} 
\rightarrow \sqrt{\frac{(n-k)!}{n!}}\sum_{j_1 \neq \dots \neq j_k}^{n-1} \ket{j_1} \dots \ket{j_k}\ket{e_{j_1} \oplus \dots \oplus e_{j_k}}.$$
\end{lemma}
\begin{proof}
First note that the state produced by the \textbf{Filling} step,
$$\frac{1}{n^{k/2}}\sum_{j_1\dots j_k = 0}^{n-1}\ket{j_1}\dots \ket{j_k} \ket{e_{j_1} \oplus \dots \oplus e_{j_k}}_{s_1} ,$$
contains states in which some of the indices $j_l$ overlap. Let $\ket{\psi} = \sum_{j_1 \neq \dots \neq j_k } \ket{j_1}\dots \ket{j_k} \ket{e_{j_1} \oplus \dots \oplus e_{j_k}}$, be the state in which none of the indices overlap, the desired output state. Then we can write
$$\frac{1}{n^{k/2}}\sum_{j_1\dots j_k = 0}^{n-1}\ket{j_1}\dots \ket{j_k} \ket{e_{j_1} \oplus \dots \oplus e_{j_k}}_{s_1} = \alpha \ket{\psi} + \beta \ket{\psi^{\perp}},$$
with $\ket{\psi^\perp}$ containing the states in which at least two of the indices $j_l$ ovelap. Note that $\braket{\psi}{\psi^\bot} = 0$, so $\alpha$ can be exactly calculated by counting the number of quantum states with distinct $j_i$'s, which gives $|\alpha|^2 = \frac{n!}{(n-k)! n^k}$.
Lemma~\ref{lem:birthday_paradox} gives a lower bound on $|\alpha|^2$: 
$$|\alpha|^2 = \frac{n!}{(n-k)! n^k} > e^{\frac{-2k^2}{n}},$$
which is at least constant for $k = \mathO(\sqrt{n})$. 

The state $\ket{\psi^{\bot}}$ is a superposition of states in which the system register state has Hamming weight less than $k$, because at least two of the $j_i$'s are the same causing a cancellation in the system register. We can use this to create a unitary $U_{flag}$ that flags $\ket{\psi^{\bot}}$. We implement this in two steps:
\begin{align*}
&\frac{1}{n^{k/2}}\sum_{j_1,\dots, j_k = 0}^{n-1}\ket{j_1}\dots \ket{j_k} \ket{e_{j_1} \oplus \dots \oplus e_{j_k}}_{s_1}\ket{0}_{\log(k)}\ket{0} \\
&\xrightarrow[]{(1)} \frac{1}{n^{k/2}}\sum_{j_1,\dots, j_k = 0}^{n-1}\ket{j_1}\dots \ket{j_k} \ket{e_{j_1} \oplus \dots \oplus e_{j_k}}_{s_1}\ket{|e_{j_1} \oplus \dots \oplus e_{j_k}|}\ket{0}\\
&\xrightarrow[]{(2)} \frac{1}{n^{k/2}}\sum_{j_1,\dots, j_k = 0}^{n-1}\ket{j_1}\dots \ket{j_k} \ket{e_{j_1} \oplus \dots \oplus e_{j_k}}_{s_1}\ket{0}\ket{\mathbbm{1}_{|e_{j_1} \oplus \dots \oplus e_{j_k}|=k}} \\
& = \alpha\ket{\psi}\ket{1} + \beta \ket{\psi^{\bot}}\ket{0}
\end{align*}
Where $|x|$ denotes the Hamming weight of bitstring $x$.
Step (1) follows from a Hamming-weight gate (see Table~\ref{tab:QFT_Hammingweight_Threshold}), which requires $\mathO(n \log(n))$ qubits; Step (2) follows from applying an Exact$_k$ gate, requiring $\mathO(\log(n)^2)$ qubits. This same step also uncomputes the Hamming-weight gate of the first step.

Lemma~\ref{lem:grover_constant_fraction} now allows us to amplify $\alpha$ to $1$ using the oracle $U_{flag}$. This produces the state
$$\sqrt{\frac{(n-k)!}{(n)!}}\sum_{j_1 \neq \dots \neq j_k} \ket{j_1} \dots \ket{j_k} \ket{e_{j_1} \oplus \dots \oplus e_{j_k}},$$
using $\mathO(k n\log(n))$ qubits.
\end{proof}

To uncompute the index registers, we have to know which one in the system register corresponds to which index register, as any permutation of the index registers results in the same state in the system register. 
The \textbf{Ordering} step imposes an ordering on the index registers, thereby removing the redundancy in the ordering. 

\begin{lemma}
\label{lem:dicke_ordering}
An $\LAQCC$ circuit exists that implements \textbf{Ordering}:
$$\sqrt{\frac{(n-k)!}{n!}}\sum_{j_1 \neq \dots \neq j_k}^{n-1} \ket{j_1} \dots \ket{j_k}\ket{e_{j_1} \oplus \dots \oplus e_{j_k}} \rightarrow  \frac{1}{\sqrt{\binom{n}{k}}}\sum_{j_1 < \dots < j_k}^{n-1} \ket{j_1} \dots \ket{j_k}\ket{e_{j_1} \oplus \dots \oplus e_{j_k}}.$$
This circuit uses $\mathO(k^2 \log(n)^2)$ qubits.
\end{lemma}
\begin{proof}
The first step of the $\LAQCC$ circuit that implements \textbf{Ordering} is to evaluate a Greaterthan-gate on all pairs of index registers, which requires $k$ copies of each index register. We require $k$ extra qubits per index register to store the outcome of the Greaterthan-gates. 
The copies of the index registers are created by doing a fan-out gate. Note that the distribution of the index registers should be set up in such a way that every possible pair can be compared by a Greaterthan-gate. 
\begin{align*}
    & \sqrt{\frac{(n-k)!}{n!}}\sum_{j_1 \neq \dots \neq j_k} \ket{j_1}^{\otimes k} \ket{0}^{\otimes k} \dots \ket{j_k}^{\otimes k} \ket{0}^{\otimes k} \ket{e_{j_1} \oplus \dots \oplus e_{j_k}} \xrightarrow[]{(1)} \\
    & \sqrt{\frac{(n-k)!}{n!}} \sum_{j_1 \neq \dots \neq j_k} 
    \big[\ket{j_1}^{\otimes k}  \ket{\mathbbm{1}_{j_1 > j_2}} \dots \ket{\mathbbm{1}_{j_1 > j_k}}\big]
    \dots \big[ \ket{j_k}^{\otimes k}  \ket{\mathbbm{1}_{j_k > j_1}} \dots \ket{\mathbbm{1}_{j_k > j_{k-1}}}\big] 
    \ket{e_{j_1} \oplus \dots \oplus e_{j_k}}.
\end{align*}
Each $\mathbbm{1}_{j_k > j_{k'}}$ is an indicator variable that evaluates to one if and only if $j_k > j_{k'}$. This step requires $\mathO(k^2\log(n)^2)$ qubits.

Next, we compute and measure the Hamming weight of the ancilla qubits $\ket{\mathbbm{1}_{j_1 > j_2}} \dots \ket{\mathbbm{1}_{j_1 > j_k}}$, using the Hamming-weight gate. 
We measure the calculated Hamming weights. 
As all index registers were distinct before measuring, these measurements directly impose an ordering on the index registers. 

\begin{align*}
&\xrightarrow[]{(\mathrm{Hamming weight})}\sqrt{\frac{(n-k)!}{n!}}\sum_{j_1 \neq \dots \neq j_k} \big[\ket{j_1} \ket{\mathbbm{1}_{j_1 > j_2} + \dots + \mathbbm{1}_{j_1 > j_k}}\big] \\
& \qquad\qquad\qquad\qquad\quad \dots \big[ \ket{j_k} \ket{\mathbbm{1}_{j_k > j_1} + \dots + \mathbbm{1}_{j_k > j_{k-1}}}\big] \ket{e_{j_1} \oplus \dots \oplus e_{j_k}}\\
&\xrightarrow[]{(measure)} \binom{n}{k}^{-1/2} \sum_{j_1 < \dots < j_k} \big[\ket{j_1} \ket{0}\big] \dots \big[ \ket{j_k} \ket{k}\big] \ket{e_{j_1} \oplus \dots \oplus e_{j_k}}
\end{align*}
This step costs $\mathO(k^2 \log(k))$ qubits.
Assume without loss of generality that the measurement outcomes impose the ordering ${j_1 < \dots < j_k}$. 
Otherwise, a permutation of the index registers achieves the same ordering, using the Permutation gate from Table~\ref{tab:Fanout_Perm}. 

Uncomputing the Hamming weights and the Greaterthan-gates gives the state 
$$\binom{n}{k}^{-1/2} \sum_{j_1 < \dots < j_k} \big[\ket{j_1} \dots \ket{j_k} \big] \ket{e_{j_1} \oplus \dots \oplus e_{j_k}}.$$
\end{proof}

The \textbf{Cleaning} step cleans the index registers for the Dicke state in a similar fashion as in the \textbf{Compress} method in the $W$-state protocol. 
In the cleaning process, we have to take the added ordering of the index registers into account. 
Suppose the $l$-th qubit of the system register is a $1$.
If this is the first $1$ in the system register, it belongs to index register $j_1$, and if it is the $m$-th $1$ it belongs to index register $j_m$. 
Computing the Hamming weight of the first $l-1$ qubits gives precisely this information.
Combined, this shows that if the $l$-th qubit is a $1$ and the Hamming weight of the first $l-1$ qubits equals $m$, then the $l$-th qubit should uncompute the $m+1$-th index register $j_{m+1}$. 

\begin{lemma}
\label{lem:dicke_cleaning}
An $\LAQCC$ circuit exists that implements \textbf{Cleaning}:
$$\frac{1}{\sqrt{\binom{n}{k}}}\sum_{j_1 < \dots < j_k}^{n-1} \ket{j_1} \dots \ket{j_k}\ket{e_{j_1} \oplus \dots \oplus e_{j_k}} \rightarrow \frac{1}{\sqrt{\binom{n}{k}}}\sum_{j_1 < \dots < j_k}^{n-1} \ket{0} \dots \ket{0}\ket{e_{j_1} \oplus \dots \oplus e_{j_k}}.$$
This circuit uses $\mathO(n^2 \log(n))$ qubits.
\end{lemma}
\begin{proof}
The first step, as described above, is to acquire the Hamming weight from all the substrings of the system register. This requires $n$ copies of the system register as well as a $\log(k)$-qubit register to store the Hamming weight value. The copies follow from the fanout-gate.
\begin{align*}
    &\binom{n}{k}^{-1/2} \sum_{j_1 < \dots < j_k} \ket{j_1} \dots \ket{j_k} \ket{e_{j_1} \oplus \dots \oplus e_{j_k}}\ket{0}_n^{\otimes n - 1}\ket{0}_{\log(n)}^{\otimes n}\xrightarrow{(1)} \\ 
    &\binom{n}{k}^{-1/2} \sum_{j_1 < \dots < j_k} \ket{j_1} \dots \ket{j_k} \ket{e_{j_1} \oplus \dots \oplus e_{j_k}}^{\otimes n}\ket{0}_{\log(n)}^{\otimes n} \xrightarrow{(2)}\\
    &\binom{n}{k}^{-1/2} \sum_{j_1 < \dots < j_k} \ket{j_1} \dots \ket{j_k} \ket{e_{j_1} \oplus \dots \oplus e_{j_k}}^{\otimes n}\bigotimes_{l=0}^{n-1}\ket{|(e_{j_1} \oplus \dots \oplus e_{j_k})_{[l,n]}|},
\end{align*}
where $|(e_{j_1} \oplus \dots \oplus e_{j_k})_{[l,n]}|$ denotes the Hamming weight of the substring consisting of qubits $l$ up until $n$ of the system register.
Step (1) copies the system qubits using fan-out gates; 
Step (2) computes the Hamming weight of all the qubits $1$ up until $j-1$ using the Hammingweight-gate shown in Table~\ref{tab:QFT_Hammingweight_Threshold};
Step (3) cleans the copies of the system register by applying fan-out. 
This step is omitted from the equations, but is included in the graphical explanation of the circuit, shown in Figure~\ref{fig:Hammingweight} for $n=4$.
Note that at the end of the calculation, it is convenient to teleport the Hamming weight registers next to the system register. There are now $n$ new registers containing the information of the Hamming weight, we will refer to them as the Hamming weight registers.
This step requires $\mathO(n^2 \log(n))$ qubits. 
\begin{figure}[ht]
\includegraphics[width=\textwidth]{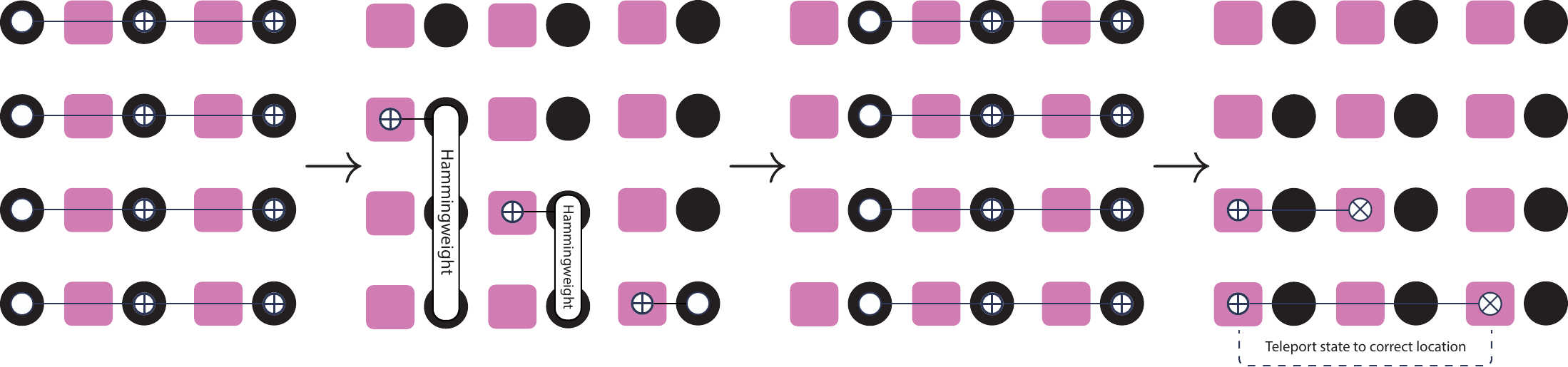}
\caption{Circuit to implement of the Hamming weight calculation of all qubit strings $l$ to $n-1$ in four steps in parallel. 
The black dots represent system qubits, the pink squares represent $\log(k)$ qubit registers that can count the Hamming weight up until $k$.}
\label{fig:Hammingweight}
\end{figure}

The last step that remains is to clean the $k$ index registers. Cleaning the $k$ index registers follows similar steps as the \textbf{Compress} method in the $W$-state protocol, with the added Hamming-weight information taken into account. This step requires $k$ copies of the system registers well as $k$ copies of the Hamming-weight registers. Every index register is paired with one copy of the system register and a copy of the $n$ Hamming-weight registers.
Cleaning the $j$-th register consists of five steps, similar to the \textbf{Compress} method of the $W$-state:
Step (1) applies Hadamard gates to bring the index register to phase space, in which $CNOT$-gates are diagonalized;
Step (2) copies the index register;
Step (3) uses the information in the Hamming-weight and system register to apply the phases to the correct index register qubits;
Step (4) cleans the index register copies;
and, Step (5) applies Hadamard gates to reset the index register qubits to the $\ket{0}$ state

Figure~\ref{fig:compress_dicke} shows the steps taken to clean a single index register $j$. 
The black dots represent the qubits in the system register and the upper row of blue dots represent the qubits in index register $j$.  The pink squares represent the ancilla Hamming weight register, where each square represents a group of $\log(k)$ qubits. This step requires $\mathO(n k \log(k)\log(n))$ qubits. At the end of the \textbf{Cleaning} operation the state is as desired:
$$
 \frac{1}{\sqrt{\binom{n}{k}}}\sum_{j_1 < \dots < j_k}^{n-1} \ket{0} \dots \ket{0}\ket{e_{j_1} \oplus \dots \oplus e_{j_k}}.
$$
The \textbf{Cleaning} step requires $\mathO(n^2 \log(n))$ qubits.
\end{proof}

\begin{figure}[ht]
\includegraphics[width=\textwidth]{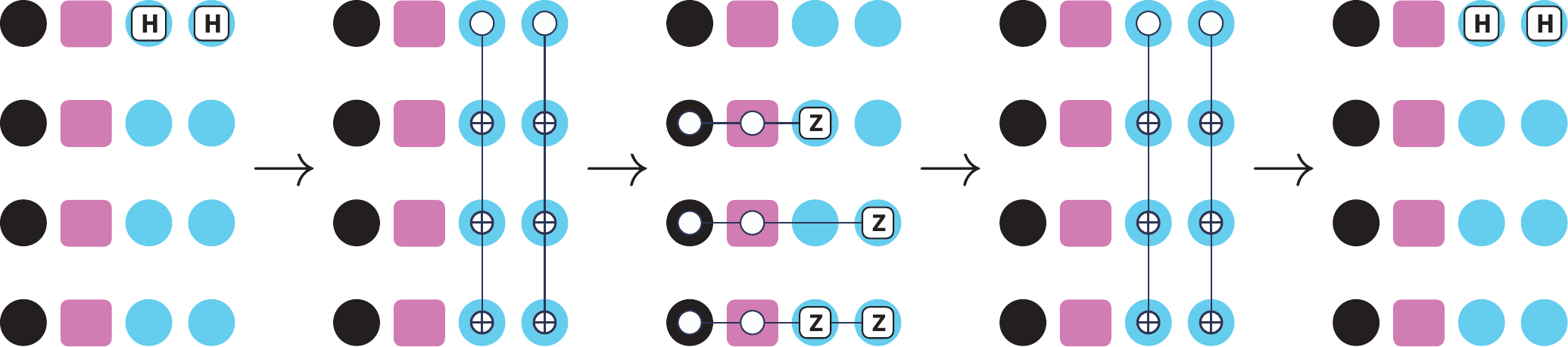}
\caption{Circuit to clean index register $j$. 
The black dots represent qubits in the system register and the blue dots the index register and its copies.
The pink squares represent the ancilla Hamming weight register and its copies. 
Each pink square represents a group of $\log(k)$ qubits.
Each of the five grids represents a single timeslice in the circuit.}
\label{fig:compress_dicke}
\end{figure}

\begin{theorem}
\label{thm:dicke_cnst_depth}
For any $n$ and $k = \mathO(\sqrt{n})$ there exists an $\LAQCC$ circuit preparing the Dicke-$(n,k)$ state, $\ket{D^n_k}$, using $\mathO(n^2\log(n))$ qubits.
\end{theorem}

\begin{proof}
The circuit combines the circuits resulting from Lemmas~\ref{lem:Dicke_filling}, \ref{lem:Dicke_filtering}, \ref{lem:dicke_ordering} and \ref{lem:dicke_cleaning}.
It consists of four steps:
\begin{align*}
\ket{0}_{i_1}\dots\ket{0}_{i_k}\ket{0}_{s_1} &\xrightarrow{(1)} \frac{1}{n^{k/2}}\sum_{j_1, \dots, j_k = 0}^{n-1} \ket{j_1}_{i_1} \dots \ket{j_k}_{i_k}\ket{e_{j_1} \oplus \dots \oplus e_{j_k} }_{s_1}\\
&\xrightarrow{(2)} \sqrt{\frac{(n-k)!}{n!}}\sum_{j_1 \neq \dots \neq j_k}^{n-1} \ket{j_1} \dots \ket{j_k}\ket{e_{j_1} \oplus \dots \oplus e_{j_k}}\\
&\xrightarrow{(3)} \frac{1}{\sqrt{\binom{n}{k}}}\sum_{j_1 < \dots < j_k}^{n-1} \ket{j_1} \dots \ket{j_k}\ket{e_{j_1} \oplus \dots \oplus e_{j_k}}\\
&\xrightarrow{(4)} \frac{1}{\sqrt{\binom{n}{k}}}\sum_{j_1 < \dots < j_k}^{n-1} \ket{0} \dots \ket{0}\ket{e_{j_1} \oplus \dots \oplus e_{j_k}}
\end{align*}
Step (1) implements \textbf{Filling} using Lemma~\ref{lem:Dicke_filling} requiring $\mathO(k n \log(n))$ qubits; 
Step (2) implements \textbf{Filtering} using Lemma~\ref{lem:Dicke_filtering} requiring $\mathO(k n \log(n))$ qubits;
Step (3) implements \textbf{Ordering} using Lemma~\ref{lem:dicke_ordering} requiring $\mathO(k^2\log(n)^2)$ qubits; 
Step (4) implements \textbf{Cleaning} using Lemma~\ref{lem:dicke_cleaning} requiring $\mathO(n^2\log(n))$ qubits. After every step ancilla qubits are cleaned, so that they can be reused. As $k=\mathO(\sqrt{n})$ the largest amount of qubits required for a step is Step (4) requiring $\mathO(n^2\log(n))$ qubits.
\end{proof}
\citeauthor{bartschi2022deterministic_short_depth} posed a conjecture on the optimal depth of quantum circuits that prepare the Dicke-$(n,k)$ state. They give an algorithm for generating Dicke-$(n,k)$ states in depth $\mathO(k \log(\frac{n}{k}))$, given all-to-all connectivity, and conjecture that this scaling is optimal when $k$ is constant. Our result shows that there is a $\LAQCC$ implementation in this regime, when one has access to intermediate measurements and feed forward. This does not disprove their conjecture. However the circuits shown here are also accessible in $\QNC^1$ by Lemma~\ref{lem:LAQCC_QNC1}, giving ``pure" quantum circuits with depth $\mathO(\log(n))$ for $k = \mathO(\sqrt{n})$ and achieving better scaling when $k = \omega(1)$. 

\subsection{Dicke states for all \texorpdfstring{$k$}{k} using log-depth quantum circuits}
\label{sec:Dicke_in_LAQCC_LOG}
The previous section gave a constant-depth protocol to prepare the Dicke-$(n,k)$ state for $k=\mathO(\sqrt{n})$. 
We developed a different method for creating Dicke-$(n,k)$ states which requires logarithmic (in $n$) quantum depth to prepare Dicke-$(n,k)$, but works for arbitrary $k$. We first define what we mean with logarithmic quantum depth:
\begin{notation}
We let $\LAQCC\text{-}\mathsf{LOG}$ refer to the instance $\LAQCC(\QNC^0,\NC^1, \mathO(\log(n)))$, similar to $\LAQCC$ except that we allow for a logarithmic number of alterations between quantum and classical calculations. This results in a circuit of logarithmic quantum depth.
\end{notation}
In this section we show a $\LAQCC\text{-}\mathsf{LOG}$ circuit that creates the Dicke-$(n,k)$ state.

One way of studying the creation of Dicke states is by looking at efficient algorithms that convert numbers from one representation to another. 
An example of this is the \textbf{Uncompress}-\textbf{Compress} method in the $W$-state protocol, that converts numbers from a binary representation to a one-hot representation. 
Dicke states are a generalization of the $W$-state, hence the one-hot representation no longer suffices for preparing the state. 
Instead, we use a construction based on number conversion between the combinatorial representation and the factoradic representation.
Below we introduce both representations and present quantum circuits that map between the two. Theorem~\ref{thm:Dicke:Log_depth} proves that a $\LAQCC\text{-}\mathsf{LOG}$ circuit can prepare the Dicke-$(n,k)$ state for any $k$. 

\subsubsection{Combinatorial number system}
An interesting result showed that any integer $m\ge 0$ can be written as a sum of $k$ binomial coefficients~\cite{Beckenbach:1964}. 
For fixed $k$, this is even unique as the next lemma shows.
\begin{lemma}[\cite{Beckenbach:1964}]
\label{lem:comb_numbers}
For all integers $m\geq 0$ and $k \geq 1$, there exists a unique decreasing sequence of integers $c_k, c_{k-1},\dots, c_1$ with $c_j > c_{j-1}$ and $c_1 \geq 0$ such that 
$$m = \binom{c_k}{k}  + \binom{c_{k-1}}{k-1} \dots \binom{c_1}{1} = \sum_{i=1}^k \binom{c_i}{i}.$$
\end{lemma}

This lemma allows for the definition of the combinatorial number representation:
\begin{definition}
Let $k \in \mathbb{N}$ be a constant. Any integer $m \in \mathbb{N}$ can be represent by a unique string of numbers $(c_k, c_{k-1} \dots, c_1)$, such that $c_k > c_{k-1} \dots >  c_1 \geq 0$ and $c_k \leq m$. 
This string is given by the unique decreasing sequence of Lemma~\ref{lem:comb_numbers}. 
We call this string the \emph{index representation} denoted by $m^{indx(k)}$.

The bit string of $k$ ones at indices $(c_k,\dots, c_1)$ is the $m$-th bit string with $k$ ones in the lexicographical order. 
This bit string is called the \emph{combinatatorial representation}.
We denote the $m$-th bit string with $k$ ones as $m^{comb(k)}$.
\end{definition}

The $W$-state protocol used the conversion between the binary representation of a number $m$ and its combinatorial representation $m^{comb(1)}$.
A generalized number conversion is precisely the protocol needed to prepare Dicke states. 

A sketch of the protocol would be as follows: 
given positive integers $k$ and $n$:
Create a superposition state 
$${\binom{n}{k}}^{-\frac{1}{2}}\sum_{i=0}^{\binom{n}{k} - 1}\ket{i}\ket{0};$$
Use number conversion to go from label $i$ to $i^{comb(k)}$
$${\binom{n}{k}}^{-\frac{1}{2}}\sum_{i=0}^{\binom{n}{k} - 1}\ket{i}\ket{i^{comb(k)}};$$
Use number conversion from $i^{comb(k)}$ to $i$ to clean up the label register
$${\binom{n}{k}}^{-\frac{1}{2}}\sum_{i=0}^{\binom{n}{k} - 1}\ket{0}\ket{i^{comb(k)}} = \ket{D^n_k}.$$

The conversion map from the combinatorial representation to the binary representation is given by Lemma~\ref{lem:comb_numbers}. 
This calculation requires iterative multiplication and addition, both of which are in $\TC^0$, hence this calculation is in $\TC^0$. 

The converse mapping, from binary to combinatorial representation for given $k$, can be achieved by a greedy iterative algorithm:
On input $m$, find the biggest $c_k$ such that $m \geq \binom{c_k}{k}$ and subtract this from $m$: $\tilde{m} = m - \binom{c_k}{k}$. 
This gives $c_k$ and a residual $\tilde{m}$. 
Repeat this process for $\tilde{m}$: 
Find the largest $c_{j}$ such that $\tilde{m}\geq \binom{c_j}{j}$ and update residual $\tilde{m} = \tilde{m} - \binom{c_j}{j}$, until all $c_j$ are found. 

This greedy algorithm is inherently linear in $k$ as it requires all previously found $\{c_i\}_{i=j}^k$ to find $c_{j-1}$. 
Hence, it is not immediately obvious if and how to achieve this mapping in constant or even logarithmic depth. 

\subsubsection{Mapping between factoradic representation and combinatorial number system}
A number representation closely related to the combinatorial number representation is the \textit{factoradic representation}. 
This number system uses factorials instead of binomials to represent numbers. 

\begin{definition} 
\label{def:factoradics}
A sequence $y = (y_{n-1}, y_{n-2}, \dots, y_0)$ of integers, such that $j \geq y_j \geq 0$ is called a \textit{factoradic}, or more explicitly an \textit{$n$-factoradic}. 
The elements of an $n$-factoradic is called an \textit{$n$-digit}.
An $n$-factoradic $y$ can represent a number $m$ between $0$ and $n!-1$, in the following way
\begin{align}
\label{eqn:fact_to_int}
m = \sum_j^{n} y_j \cdot j!.
\end{align}
For a given $m \in \{0, \ldots, n!-1\}$, we call the $n$-factoradic $y$ obeying the equality above, the \emph{factoradic representation of $m$}.  
Denote $\text{Fact}(n)$ as the set of all $n$-factoradics.
\end{definition}

The following lemma shows that Equation~\ref{eqn:fact_to_int} is a bijection, showing that the factoradic representation is unique.

\begin{lemma}
\label{factoradic_summation}
For $k \geq 0$ it holds that:
$$\sum_{i=0}^k i \cdot i!= (k + 1)! - 1.$$ 
\end{lemma}
\begin{proof}
Proof by induction.\\
\textbf{BASE STEP}: Let $k$ be $0$:
$$0\cdot 0! = 1! - 1$$
\textbf{INDUCTION STEP}: Assume the lemma holds for some $j$, then
$$\sum_{i=0}^{j+1} i \cdot i! = (j+1) \cdot (j+1)! + \sum_{i = 0}^{j} i \cdot i = (j+1) \cdot (j+1)! + (j+1)! - 1 = (j+2)! - 1,$$
which completes the proof.
\end{proof}
This identity allows for using factorials as a base for a number system.
The next lemma gives a log-space algorithm to convert a factoradic representation to its combinatorial representation. 
\begin{lemma}
\label{lem:fac_to_comb}
There is a logspace algorithm $\mathcal{A}$ that, given $k \in \{0, \ldots, n\}$, and a uniformly random $n$-factoradic, outputs a uniformly random $n$-bit string of Hamming weight $k$.
\end{lemma}
\begin{proof}
The algorithm $\mathcal{A}$ is given $k$ and an $n$-factoradic $y = (y_{n-1}, \dots, y_0)$. It will then output an $n$-bit string $y^{comb(k)} = y^{comb(k)}_{n-1} \dots y^{comb(k)}_0 \in \{0,1\}^n$ of Hamming weight $k$, one bit at a time, from left to right, according to the following rule. Let $H_{>n-j} = \sum_{i=n-j+1}^{n-1} y^{comb(k)}_i$ be the Hamming weight of the bits produced before we reach bit $n-j$. Then $y^{comb(k)}_{n-j}$ is given by:
\begin{align}
\label{eqn:fac_to_comb}
    (\mathcal{A} (y))_{n-j} = y^{comb(k)}_{n-j} = \begin{cases} 1 & \text{if } y_{n-j} < k - H_{>n-j} \\
    0 & \text{otherwise}
    \end{cases}.
\end{align}
This conversion requires comparing an $n$-digit with a constant and the Hamming weight of a bitstring.
The only information that $\mathcal A$ needs to remember, as it goes from bit $n-j+1$ to bit $n-j$ , is the Hamming weight $H_{>n-j}$ of the bits it produced so far, and this can be stored in logarithmic space. 

Now note that the number of factoradic $n$-digit strings that map to the same combinatorial bit string is always $k!(n-k)!$:
Let $y^{comb(k)} \in \{0,1\}^n$ have Hamming weight $k$. 
For any bit position $y^{comb(k)}_{n-j}$, there are $n - j + 1 - (k - H_{>n-j})$ possible choices for the $n$-digit $y_{n-j} \in \{0, \ldots, n-j\}$ that set $y^{comb(k)}_{n-j} = 0$. 
For the leftmost index $n-j$ such that $y^{comb(k)}_{n-j} = 0$, it holds that $H_{>n-j} = j-1$, and then there are $n-k$ possible $n$-digits $y_{n-j}$ that set $y^{comb(k)}_{n-j} = 0$. 
Then, for the second index $n-j$ such that $y^{comb(k)}_{n-j} = 0$ it holds that $H_{>n-j} = j - 2$, hence there are $n - k - 1$ possible $n$-digits $y_{n-j}$ causing $y^{comb(k)}_{n-j} = 0$. And so forth. 
This results in $(n-k)!$ different possible choices for the $(n-k)$-many $n$-digits where $y^{comb(k)}=0$. 

Similarly, for the leftmost position $n-j$ where $y^{comb(k)}_{n-j} = 1$, there are $k$ possible choices for the $n$-digit $y_{n-j}$ that cause $y^{comb(k)}_{n-j} = 1$. 
The second leftmost position $n-j$ gives $k-1$ possible choices, and so forth, for a total of $k!$ possible settings of the $k$-many $n$-digits where $y^{comb(k)}=0$.

Combined, we conclude that, for every $n$-bit string $y^{comb(k)} \in \{0,1\}^n$ of Hamming weight $k$, there are exactly (the same number of) $k!(n-k)!$ $n$-factoradics $y$ such that $\mathcal A(y) = y^{comb(k)}$.
Hence, a uniformly random $n$-factoradic is mapped by $\mathcal A$ to a uniformly random $n$-bit string of Hamming weight $k$, as claimed.
\end{proof}

This lemma gives a logspace algorithm to convert a uniformly random $n$-factoradic to a uniformly random $n$-bit string of Hamming weight $k$, for any $k$.
It is well known that logspace is contained in $\TC^1$, allowing this calculation to be performed in parallel log-depth when one has access to threshold gates~\cite{Johnson:1990}. As we saw in Section~\ref{sec:gates_created_in_LAQCC}, we can compute a threshold gate in $\LAQCC$. Hence, an $\LAQCC\text{-}\mathsf{LOG}$ can perform any $\TC^1$ calculation. We conclude:

\begin{corollary}\label{cor:fac_to_comb}
The following map can be implemented in $\LAQCC\text{-}\mathsf{LOG}$.
\[
\frac{1}{\sqrt{n!}} \sum_{y \in \text{Fact}(n)} \ket{y}\ket{0} \xrightarrow{} \frac{1}{\sqrt{n!}} \sum_{y \in \text{Fact}(n)} \ket{y}\ket{\mathcal A (y)}.
\]
\end{corollary}

\noindent In the next lemma, we show that a $\TC^0$ circuit can implement the inverse of $\mathcal A$. 
\begin{lemma}
\label{lem:comb_to_fac}
There exists a $\TC^0$ algorithm which, when given an $n$-bit string $y^{comb(k)}$ of Hamming weight $k$, a uniformly-random $k$-factoradic, and a uniformly-random $(n-k)$-factoradic, outputs a uniformly random $n$-factoradic $y$ among those such that $\mathcal A(y) = y^{comb(k)}$.
\end{lemma}

\begin{proof}
The conversion can be done in parallel, generating an $n$-digit for every bit in $y^{comb(k)} = y_{n-1} \dots y_0\in\{0,1\}^n$. Recall that we are given as input a uniformly-random $k$-factoradic $O_{k-1}, \dots, O_0$ and a uniformly-random $(n-k)$-factoradic $Z_{n-k-1}, \dots, Z_0$.

For every bit position $n-j$, for $1 \le j \le n$, calculate the Hamming weight of the bits from $n-j+1$ to $n-1$: $H_{>n-j} = \sum_{i=j+1}^{n-1} y^{comb(k)}_i$. Recall that iterated addition is in $\TC^0$~\cite{vollmer1999introduction}.

If $y^{comb(k)}_{n-j} = 1$, set $y_{n-j} = O_{k - H_{> n-j}}$. This gives a uniform random $n$-digit between $0$ and $k - H_{> n-j} - 1$. If $y^{comb(k)}_{n-j} = 0$, set $y_{n-j} = k - H_{> n-j} + Z_{n-k-H_{>n-j}}$. 
Note that this gives a uniform random $n$-digit between $k - H_{>n-j}$ and $n - j$. 
By construction, it now follows that $\mathcal A(y) = y^{comb(k)}$. 
Computing each $n$-digit in this way requires summation and indexing, both of which are in $\AC^0 \subseteq \TC^0$~\cite{vollmer1999introduction}.
\end{proof}

\begin{remark}\label{rem:comb_to_fac}
The above algorithm establishes a bijection $(y^{comb(k)}, Z, O) \leftrightarrow y$  between triples $(y^{comb(k)}, Z, O)$ with $y^{comb(k)} \in \{0,1\}^n$ of Hamming weight $k$, $Z \in \text{Fact}(n-k)$ and $O\in\text{Fact}(k)$ and $n$-factoradics $y \in \text{Fact}(n)$. Let $(\mathcal A(y), \mathcal Z(y), \mathcal O(y))$ be the image of an $n$-factoradic $y$ under this bijection. The previous lemma shows that one can compute $y$ from $(y^{comb(k)}, Z, O)$ in $\TC^0$. 
It is not hard to see that the map $(\mathcal A(y), y) \mapsto (\mathcal A(y), y, \mathcal Z(y), \mathcal O(y))$ is also in $\TC^0$. Indeed, to find $\mathcal Z(y)$ and $\mathcal O(y)$, we need only invert the two defining equalities $y_{n-j} = O_{k - H_{> n-j}}$ and $y_{n-j} = k - H_{> n-j} + Z_{n-k-H_{>n-j}}$.
\end{remark}

\begin{corollary}\label{cor:comb_to_fac}
The following map can be implemented in $\LAQCC$.
\[
\begin{pmatrix}n \\ k\end{pmatrix}^{-\frac{1}{2}}\sum_{y^{comb(k)}} \ket{0} \ket{y^{comb(k)}} \xrightarrow{} \frac{1}{\sqrt{n!}} \sum_{y \in \text{Fact}(n)} \ket{y}\ket{\mathcal A (y)}
\]
where $y^{comb(k)}$ ranges over all $n$-bit strings of Hamming weight $k$.
\end{corollary}

\begin{proof}
The transformation consists of three steps: 
\begin{align*}
& \binom{n}{k}^{-\frac{1}{2}}\sum_{y^{comb(k)}} \ket{y^{comb(k)}} \ket{0}\ket{0}\ket{0}\\
\xrightarrow{(1)} \;\; &\binom{n}{k}^{-\frac{1}{2}}\sum_{y^{comb(k)}} \ket{y^{comb(k)}} \left(\bigotimes_{j = 0}^{n-k-1} \sum_{i = 0}^{j} \ket{i}\right)\left(\bigotimes_{j = 0}^{k-1} \sum_{i = 0}^{j} \ket{i}\right)\ket{0}\\
= \;\; & \frac{1}{\sqrt{n!}} \sum_{y^{comb(k)}} \ket{y^{comb(k)}} \left(\sum_{Z\in\text{Fact}(n-k)} \ket{Z}\right)\left(\sum_{O\in\text{Fact}(k)} \ket{O}\right)\ket{0}\\
\xrightarrow{(2)}\;\; & \frac{1}{\sqrt{n!}} \sum_{y\in\text{Fact}(n)} \ket{\mathcal A(y)} \ket{\hat Z(y)}\ket{\hat O(y)}\ket{y}\\
\xrightarrow{(3)}\;\; & \frac{1}{\sqrt{n!}} \sum_{y\in\text{Fact}(n)} \ket{\mathcal A(y)} \ket{0}\ket{0}\ket{y}
\end{align*}
Step (1) prepares a uniform superposition over all $n$-factoradics using Theorem~\ref{thm:uniform_superposition_mod_q}. 
Step (2)  is Lemma~\ref{lem:comb_to_fac}, and Step (3) follows from Remark~\ref{rem:comb_to_fac}.
In the above steps we implicitly used that the inverse of the used $\LAQCC$ operations are also $\LAQCC$ operations. 
Even though it is unclear if this inverse-property holds in general, it does hold for the considered $\LAQCC$ operations. 
The measurement steps, which might not be reversible, in this algorithm are used to implement fan-out gates.
The inverse of a fan-out gate is the fan-out gate itself and hence is contained in $\LAQCC$.
\end{proof}

\begin{theorem}
\label{thm:Dicke:Log_depth}
There exists a $\LAQCC\text{-}\mathsf{LOG}$-circuit for preparing Dicke-$(n,k)$ states, for any positive integers $n$ and $k \le n$, it uses $\mathO(\text{poly}(n))$ qubits. 
\end{theorem}
\begin{proof}
The circuit combines the circuits resulting from Lemma~\ref{lem:fac_to_comb} and Lemma~\ref{lem:comb_to_fac}. 

It consists of three steps: 
\begin{align*}
\ket{0}^{\otimes n \log(n)}\ket{0}^{\otimes n} &\xrightarrow{(1)} \frac{1}{\sqrt{n!}}\left(\bigotimes_{j = 0}^{n-1} \sum_{i = 0}^{j} \ket{i}\right)\ket{0}^{\otimes n} = \sum_{y \in \text{Fact}(n)} \ket{y}\ket{0} \\
&\xrightarrow{(2)} \frac{1}{\sqrt{n!}} \sum_{y \in \text{Fact}(n)} \ket{y}\ket{\mathcal A (y)}\\
&\xrightarrow{(3)} \begin{pmatrix}n \\ k\end{pmatrix}^{-\frac{1}{2}}\sum_{y \in \text{Fact}(n)} \ket{0}\ket{\mathcal A (y)} = \ket{D^n_k}.
\end{align*}
Step (1) prepares a uniform superposition over all $n$-factoradics using Theorem~\ref{thm:uniform_superposition_mod_q};
Step (2) is by Corollary \ref{cor:fac_to_comb};
and, Step (3) reverses the algorithm of Corollary \ref{cor:comb_to_fac}.
\end{proof}

\subsection{Quantum many-body scar states}
\label{sec:many_body_scar}
There is a particular set of states in many-body physics called, many-body scar states, which are highly excited states that exhibit atypically low entanglement~\cite{turner2018weak}. These states exhibit long coherence times relative to other states at the same energy density and seem to avoid thermalization and thereby they do not follow the eigenstate thermalization hypothesis.
This makes studying the lifetime of quantum many body scar states under perturbations particularly interesting.
Studying this lifetime is quite challenging, as even though scarred eigenstates often have modest entanglement and therefore have efficient matrix product state representations, perturbations typically couple them to states nearby in energy which typically have volume-law scaling entanglement, making classical simulations difficult.

A overview paper by \citeauthor{Gustafson:2023} studied methods of preparing quantum many-body scar states on quantum computers, with the goal to simulate time dynamics directly on the quantum system~\cite{Gustafson:2023}. 
They found several approaches for generating quantum many-body scars for a particular model, which require polynomial depth. They look at quantum many-body scar states of the $n$-qubit spin-1/2 Hamiltonian of~\cite{iadecola2020quantum}:
\[
    H = \lambda \sum_{i = 2}^{n - 1} (X_i - Z_{i - 1} X_i Z_{i +1}) + \Delta \sum_{i = 1}^n Z_i + J \sum_{i = 1}^{n-1} Z_i Z_{i + 1}
\]

The quantum many-body scar states of $n$-qubits are given by:
\[
\ket{S_k} = \frac{1}{k! \sqrt{\mathcal N(n,k)}}(Q^\dagger)^k \ket{\Omega},
\]
where $\mathcal N(n,k) = \binom{n - k - 1}{k}$, $\ket{\Omega} = \ket{0}^{\otimes n}$ and $k = 0, \dots, n/2 - 1$. The raising operator $Q^\dagger$ is given by:
\[
 Q^\dagger = \sum_{i = 2}^{n-1} (-1)^i P_{i-1} \sigma^+_{i} P_{i +1}, 
\]
with $P_j = \ket{0}\bra{0}$ and $\sigma^+_j = X_j + Y_j$.
They show that up to local $Z$ gates these states are very closely related to Dicke states:
\[
    \prod_{\text{i odd}} Z_i \ket{S_k} = \ket{0} \otimes P_{fib} \ket{D^n_k}\otimes \ket{0},
\]
where $P_{fib}$ is known as the Fibonacci constraint, which is a projector that removes all states where there are two ones next to each other:
\[
P_{fib} = I - \sum_{i = 1}^{n-1} \ket{11}\bra{11}_{i,i+1}.
\]
The goal of this section is to show that these states, for $k = \mathO(\sqrt{n})$ are accessible in $\LAQCC$. First note that by Theorem~\ref{thm:dicke_cnst_depth} there exists a $\LAQCC$ protocol to generate $\ket{D_k^n}$ up to $k = \mathO(\sqrt{n})$.
We will show that there is a $\LAQCC$ protocol that applies $P_{fib}$ to these $\ket{D^n_k}$ states. The first step will be to show that there exists a unitary accessible in $\LAQCC$ that flags the correct state.

\begin{lemma}
\label{lem:Ufib}
There exists a unitary $U_{fib}$, accessible in $\LAQCC$, that flags all the states that obey the Fibonacci constraint, more precisely:
\[
    U_{fib} \ket{D_k^n}\ket{0} = \alpha P_{fib}\ket{D_k^n}\ket{0} + \beta (I - P_{fib})\ket{D_k^n}\ket{1}
\]
\end{lemma}
\begin{proof}
Add $n$ extra qubits prepared in $\ket{0}$, one for every sequential pair of qubits. 
For all $i\in\{1,\hdots,n-1\}$, apply a Toffoli gate with control qubits $i$ and $i + 1$ and target qubits the $i$-th auxillary qubit. 
The second step is to apply the $OR_{n}$ gate on the $n$ auxillary qubits.
The $n$-th auxillary qubit is also used as the output qubit. 
Clean the extra qubits by again applying Toffoli gates. These steps are accessible in $\LAQCC$ therefore implements the flag unitary with an $\LAQCC$ protocol.
\end{proof}

The second step is to show that $\alpha$ is bounded by a constant in the case that $k = \mathO(\sqrt{n})$, this is shown using the following two lemma's.

\begin{lemma}

The total number of bitstrings of length $n$ with $k$ ones, such that no two ones are adjacent is given by:
\[
    \binom{n - k}{k} + \binom{n - k - 1}{k - 1}
\]
\end{lemma}

\begin{proof}
We can count the number of possible bitstrings, after first noticing that every $1$ must be followed by a $0$, unless the last bit is a $1$. 
As a result, we have two situations, in the first, we can consider all possible rearrangements $n-k$ elements, consisting of $k$ pairs `$10$' and $n-2k$ ones. 
This gives $\binom{n-k}{k}$ possible bitstrings. 
In the second situation, the last bit is $1$. 
This leave $k-1$ pairs `$10$' in a total of $n-k-1$ elements. 
With the same reasoning, this gives $\binom{n-k-1}{k-1}$ possible bitstrings. 
Summing the two situations proves the lemma. 
\end{proof}

We now consider the relative fraction of this type of bitstrings among all possible bitstrings with Hamming-weight $k$. 
\begin{lemma}
\label{lem:fraction_good_strings}
Let $k = c\sqrt{n}$ for some constant $c>0$. 
Then the following inequality holds
\[
\frac{\binom{n - k}{k} + \binom{n - k - 1}{k - 1}}{\binom{n}{k}} \geq \exp\big(-c^2\big)
\]
\end{lemma}

\begin{proof}
We have
\[
\frac{\binom{n - k}{k} + \binom{n - k - 1}{k - 1}}{\binom{n}{k}} > \frac{\binom{n - k}{k}}{\binom{n}{k}} = \frac{(n-k)! / k!(n-2k)!)}{n!/ k!(n-k!)} = \frac{(n-k)!^2}{n!(n-2k)!}.
\]
Expanding the factorials and only consider the terms that do not cancel, we obtain
\[
\frac{(n-k)!}{n!} \frac{(n-k)!}{(n-2k)!} = \frac{(n-k)(n- k - 1)\dots (n- (2k - 1))}{n (n-1)\dots (n - (k-1))}.
\]
Both the numerator and denominator have $K$ terms, which we can pair. 
next we note that $\tfrac{a}{b} > \tfrac{a-1}{b-1}$ whenever $b>a$ (and $b\not\in\{0,1\}$). 
Using this idea, we obtain the following expression:
\[
\frac{n-k}{n}\frac{n-k-1}{n-1}\dots \frac{n-(2k-1)}{n-(k-1)} > (\frac{n-k}{n})^k = (1 - \frac{k}{n})^k.
\]
Now using that $k = c \sqrt{n}$, we have
\[
\left(1 - \frac{c \sqrt{n}}{n}\right)^{c\sqrt{n}}  = \left(\left(1-\frac{c}{\sqrt{n}}\right)^{\frac{\sqrt{n}}{c}}\right)^{c^2} > \exp\big(-c^2\big),
\]
which is a constant.
\end{proof}

This allows us construct the state $\ket{S_k}$ using the steps for the Dicke-state preparation together with Lemma~\ref{lem:grover_constant_fraction}. 
Note that Lemma~\ref{lem:grover_constant_fraction} requires us to implement both $U$ and $U^{\dagger}$, where $U$ implements the initial superposition.
The \textbf{Ordering}-step (Lemma~\ref{lem:dicke_ordering}) does however use measurements, which stops us makes implementing the inverse of the circuit hard. 
Still, we can work around this, by applying Lemma~\ref{lem:grover_constant_fraction} between the \textbf{Filtering} and \textbf{Ordering} step:
\begin{theorem}
For any $n$ and $k = \mathO(\sqrt{n})$ there exists a $\LAQCC$ circuit preparing the many-body scar state $\ket{S_k}$, using $\mathO(n^2\log(n))$ qubits.
\end{theorem}
\begin{proof}
We follow the same steps as for the Dicke-state preparation (see Theorem~\ref{thm:dicke_cnst_depth}). 
After the second \textbf{Filtering} step however, we apply use the unitary $U_{fib}$ together with Lemma~\ref{lem:grover_constant_fraction} to filter out all states with subsequent ones in the state. 
Note that by Lemma~\ref{lem:fraction_good_strings}, the number of states with no consecutive ones is a constant fraction of the total number of strings of Hamming-weight~$k$. 

Furthermore, the state after the Filtering step,
\begin{equation*}
    \sqrt{\frac{(n-k)!}{(n)!}}\sum_{j_1 \neq \dots \neq j_k} \ket{j_1} \dots \ket{j_k} \ket{e_{j_1} \oplus \dots \oplus e_{j_k}},
\end{equation*}
is still entangled with the index registers. In effect there are many copies of the Dicke-$(n,k)$ state in on the system register, each with a diferent ordering of the index registers, however this does not affect the fraction of states with no consecutive ones compared to the states with consecutive ones. 
Next, the \textbf{Ordering} and \textbf{Cleaning} step of the protocol work similarly on the resulting state and will give the state $\ket{S_k}$.
\end{proof}

%% file: appendix.tex
\section{Useful lemmas}
This section gives two lemmas.
The first lemma upper bounds the computational power of $\LAQCC$.
The second lemma helps in preparing Dicke states for $k\in\mathO(\sqrt{n})$. 
\begin{lemma}
\label{lem:NC1toQauntum}
Let $\Pi = (\Pi_{yes}, \Pi_{no})$ be a decision problem in $\NC^1$. Then there is a uniform log-depth quantum circuit that decides on $\Pi$.
\end{lemma}
\begin{proof}
Let $B$ be the uniform Boolean circuit of logarithmic depth deciding on $\Pi$. 
As $\Pi\in\NC^1$, such a circuit exists. 

For fixed input size $n$, write $B$ as a Boolean tree of depth $\mathO(\log(n))$, with at its leaves the $n$ input bits $x_i$ and as root an output bit. 
This Boolean tree directly translates in a classical circuit using layers of AND, OR and NOT gates.

Each of these gates has a direct quantum equivalent gate, provided that we use ancilla qubits: 
First replace all OR gates by AND and NOT gates. 
Then replace all AND gates by Tofolli gates, which has three inputs.
The third input is a clean ancilla qubit and will store the AND of the other two inputs. 
Finally, replace all NOT gates by $X$-gates. 
\end{proof}
\begin{lemma}
\label{lem:birthday_paradox}
Let $n, k \in \mathbb{N}$ and $k < \frac{n}{2}$ then:
\begin{align*}
    \frac{n!}{n^k (n-k)!} > e^{-\frac{2 k^2}{n}}
\end{align*}
\end{lemma}
\begin{proof}
The result follows by a simple computation
\begin{align*}
    \frac{n!}{n^k (n-k)!} &= e^{\sum_{i = 1}^k \log(1 - \frac{i}{n})}\\
    & > e^{\sum_{i = 1}^k \frac{-i}{n-i}} \\
    & > e^{\sum_{i = 1}^k \frac{-i}{n-k}} \\
    & > e^{-\frac{k^2}{n-k}} \\
    & > e^{-\frac{2 k^2}{n}},
\end{align*}
where we use that $\log(1 + x)\geq \frac{x}{1+x}$ for $x > - 1$.

\end{proof}
\section{Gate implementations}

\subsection{OR-gate}\label{gate:OR_implementation}
In this section we discuss the implementation of the OR$_n$-gate and why we can implement it using local gates in a nearest-neighbor architecture.
We show how the gate works for any basis state $\ket{x}=\ket{x_1}\otimes \hdots\otimes\ket{x_n}$. 
By linearity, the gate then works for arbitrary superpositions. 
The OR-gate by \citeauthor{TakahashiTani_CollapseOfHierarchyConstantDepthExactQuantumCircuits_2013} consists of two steps: First they apply the OR-reduction introduced in Ref.~\cite{HoyerSpalek:2005}, which prepares a state on $\log n$ qubits such that the OR evaluated on these $\log n$ qubits yields the same result as the OR evaluated on the original $n$ qubits.
Second, they  evaluate an exponential circuit on these $\log n$ qubits to calculate the OR-gate. This results in a polynomial-sized circuit for the OR-gate. 

Let $m=\ceil{\log_2(n+1)}$, then the OR-reduction implements the map
\begin{equation*}
    \ket{x}\ket{0}^{\otimes m}\to \ket{x}\bigotimes_{j=1}^m\ket{\mu^{|x|}_{\phi_j}},
\end{equation*}
where $\varphi_j=\frac{2\pi}{2^j}$ and $\ket{\mu_{\varphi}^{|x|}} = H R_Z(\varphi\cdot |x|) H\ket{0}$. 
The OR-reduced state thus depends on the weighted Hamming weight of $x$,
more precisely, every $\ket{\mu^{|x|}_{\phi_j}}$ depends on the entire string $x$. For every $\ket{\mu^{|x|}_{\phi_j}}$ this requires a copy of $\ket{x}$ which can be created by using the fanout gate. The circuit applying the rotation $R_Z(\varphi\cdot |x|)$ consists of sequential $R_Z$ gates, which are diagonal and hence can be implemented in parallel  by Lemma~\ref{lem:unitar_parallelization}. This results in a circuit of width $\mathO(n \log(n))$ for creating the OR-reduced state.


We have for every bitstring $x\in\{0,1\}^n$ that 
\begin{equation*}
    \text{OR}_n(x) = \frac{1}{2^{n-1}}\sum_{a\in\{0,1\}^n\setminus \{0^n\}} \text{PA}_n^a(x),
\end{equation*}
where PA$_n^a(x) = \oplus_{j=0}^{n-1} a_i x_i$ is the parity of $x$, weighted by the non-zero binary vector~$a$.
Hence, computing the OR of the input is now reduced to computing all parities of the subsets of the inputs. 
The parity gate is equivalent to a fanout-gate conjugated by Hadamard gates on every qubit, see also Table~\ref{tab:Fanout_Perm}.

We now copy the $m$ qubits of the OR-reduced state $2^m = n$ times, using fanout gates. 
For each copy we require two additional auxiliary qubits. 
The first will hold the result of the parity computation, for which we already have a nearest-neighbor implementation. 
We will entangle the second auxiliary qubit using a fanout-gate to obtain a GHZ-state in these qubits. 

We now compute in parallel the parity of the subsets of the inputs. 
For every subset we use a single copy of the inputs and we store the result in the corresponding auxiliary qubit. 
We then apply a controlled-$R_Z$ gate from the first auxiliary qubit to the second auxiliary qubit in the GHZ state. 
This prepared the state (omitting other registers)
\begin{equation*}
    \frac{1}{\sqrt{2}}(\ket{0}^{n-1} + (-1)^{\frac{1}{2^{n-1}}\sum_{a\in\{0,1\}^m\setminus \{0^m\}} \text{PA}_m^a(x)}\ket{1}^{\otimes n-1}) = \frac{1}{\sqrt{2}}(\ket{0}^{n-1} + (-1)^{\text{OR}_{n}(x)}\ket{1}^{\otimes n-1}).
\end{equation*}
Uncomputing this final state using a fanout-gate gives a single qubit that holds the desired answer in its phase. 
A single Hadamard gate applied to this qubit will then give the answer in a single qubit. 

Combining all steps thus gives an implementation for the OR gate using a geometrically local nearest-neighbor circuit. 
For more details on the implementation as well as a proof of correctness, we refer to the original proof~\cite{TakahashiTani_CollapseOfHierarchyConstantDepthExactQuantumCircuits_2013}.

\subsection{Equality-gate}
\label{gate:equality}
Define the Equality gate on two $n$-qubit computational basis states as
$$\mathrm{Equality}: \ket{x}\ket{y}\ket{0}\mapsto\ket{x}\ket{y}\ket{\mathbbm{1}_{x=y}}.$$
This gate is implemented in three steps: 
(1) subtract the first register from the second using a subtraction circuit; 
(2) apply $\mathrm{Equal}_0$ on the second register and store the result in the third register; 
(3) add the first register to the second, undoing the subtraction computation:
\begin{align*}
    \ket{x}\ket{y}\ket{0}&\xrightarrow[(1)]{} \ket{x}\ket{y - x}\ket{0}\\
    &\xrightarrow[(2)]{} \ket{x}\ket{y - x}\ket{\mathbbm{1}_{x=y}}\\
    &\xrightarrow[(3)]{}\ket{x}\ket{y}\ket{\mathbbm{1}_{x=y}}
\end{align*}
Addition and subtraction both have width $\mathO(n^2)$, which, as a result, the $\mathrm{Equality}$-gate also has.

\subsection{Greaterthan-gate}
\label{gate:greaterthan}
Define the Greaterthan gate on two $n$-qubit computational basis states as
$$\mathrm{Greaterthan}: \ket{x}\ket{y}\ket{0}\mapsto\ket{x}\ket{y}\ket{\mathbbm{1}_{x>y}}.$$
This gate is implemented in four step: 
(1) Add an extra clean qubit to the second register and interpret this as an $n+1$-qubit register with most significant bit zero; 
(2) subtract the first register from the second. The subtraction is modulo $2^{n+1}$; 
(3) apply a CNOT-gate from most significant bit of the second register to the third register; 
(4) add the first register to the second, undoing the subtraction computation:
\begin{align*}
    \ket{x}\ket{y}\ket{0}&\xrightarrow[(1)]{} \ket{x}\ket{0y}\ket{0}\\
    &\xrightarrow[(2)]{} \ket{x}\ket{y- x \bmod 2^{n+1}}\ket{0}\\
    &\xrightarrow[(2)]{} \ket{x}\ket{y - x \bmod 2^{n+1}}\ket{\mathbbm{1}_{x>y}}\\
    &\xrightarrow[(3)]{}\ket{x}\ket{0}\ket{y}\ket{\mathbbm{1}_{x>y}}
\end{align*}
This construction works, as after step (2), the most significant bit of the second register is one, precisely if $x$ is larger than $y$. 

Addition and subtraction both have width $\mathO(n^2)$, which, as a result, the $\mathrm{Greaterthan}$-gate also has.

\subsection{Exact\texorpdfstring{$_t$}{}-gate}
\label{gate:exact_t}
Define the Exact$_t$ gate on an $n$-qubit computational basis state as
$$\mathrm{Exact}_t: \ket{x}\ket{0}\mapsto\ket{x}\ket{\mathbbm{1}_{|x|=t}}.$$
Here, $|x|$ denotes the Hamming weight of the $n$-bit string $x$. 

This gate follows by combining the Hammingweight-gate and the Equality-gate: First, compute the Hamming weight of $x$ in an ancilla register and then apply the Equality gate to check that this ancilla register equals $t$. 

Another approach is to modify the circuit for $OR$ slightly. 
In the $OR$-reduction step, add a gate $R_Z(-\varphi t)$, which adjusts the angle to be zero precisely if $|x|=t$ (see Theorem 4.6 of~\cite{HoyerSpalek:2005}).
Then apply the circuit for $OR$ and negate the output. 
The circuit for $OR$ evaluates to zero, precisely if the input had Hamming weight $t$.

\subsection{Threshold\texorpdfstring{$_t$}{}-gate}
\label{gate:threshold_t}
Define the Threshold$_t$ gate on an $n$-qubit computational basis state as
$$\mathrm{Threshold}_t: \ket{x}\ket{0}\mapsto\ket{x}\ket{\mathbbm{1}_{|x|\ge t}}.$$

Taking the $OR$ over the outputs of $\mathrm{Exact}_j$-gates for all $j\ge t$, gives the $\mathrm{Threshold}_t$-gate. 
An improved implementation with better scaling in $t$ is given in Theorem~2 of~\cite{TakahashiTani_CollapseOfHierarchyConstantDepthExactQuantumCircuits_2013}.

A weighted threshold gate uses weights $w_i$ and evaluates to one precisely if $\sum_i w_i x_i \ge t$. 
Assume without loss of generality that $\sum_i w_i x_i$ evaluates to an integer. 
Otherwise, we can use the same ideas, but up to some precision. 

Use the same $OR$-reduction as for the normal threshold gate. 
Instead of rotations $R_Z(\varphi)$ controlled by $x_i$, we use rotations $R_Z(w_i \varphi)$ controlled by $x_i$. 
This implements the weighted threshold gate.